\documentclass[letter,11pt]{article}
\usepackage[utf8]{inputenc}
\usepackage{amsthm}
\usepackage{authblk}
\usepackage{amsmath,amssymb}
\usepackage{url}
\usepackage{hyperref}
\usepackage{graphicx}
\usepackage{mathrsfs}
\usepackage[labelformat=simple]{subcaption}
\usepackage{xcolor}
\usepackage{graphv1}
\usetikzlibrary{hobby}
\usepackage[margin=1in]{geometry}
\usepackage[linesnumbered,ruled,vlined]{algorithm2e}
\newtheorem{theorem}{Theorem}
\newtheorem{lemma}{Lemma}

\newtheorem{claim}{Claim}
\newtheorem{observation}{Observation}
\newtheorem{corollary}{Corollary}
\newtheorem{definition}{Definition}
\newcommand{\interior}[1]{\overline{#1}}
\newcommand{\prob}{{\sc TS-Ind-set reconfig}}
\newcommand{\com}[1]{}
\title{Token sliding independent set reconfiguration on block graphs}
\author{Mathew C. Francis, Veena Prabhakaran}
\affil{Indian Statistical Institute, Chennai Centre\\ \href{mathew@isichennai.res.in}{mathew@isichennai.res.in},\href{veenaprabhakaran7@gmail.com}{veenaprabhakaran7@gmail.com}}

\date{}
\begin{document}

\maketitle

\begin{abstract}
Let $S$ be an independent set of a simple undirected graph $G$. Suppose that each vertex of $S$ has a token placed on it. The tokens are allowed to be moved, one at a time, by sliding along the edges of $G$, so that after each move, the vertices having tokens always form an independent set of $G$. We would like to determine whether the tokens can be eventually brought to stay on the vertices of another independent set $S'$ of $G$ in this manner. In other words, we would like to decide if we can transform $S$ into $S'$ through a sequence of steps, each of which involves substituting a vertex in the current independent set with one of its neighbours to obtain another independent set. This problem of determining if one independent set of a graph ``is reachable'' from another independent set of it is known to be PSPACE-hard even for split graphs, planar graphs, and graphs of bounded treewidth. Polynomial time algorithms have been obtained for certain graph classes like trees, interval graphs, claw-free graphs, and bipartite permutation graphs. We present a polynomial time algorithm for the problem on block graphs, which are the graphs in which every maximal 2-connected subgraph is a clique. Our algorithm is the first generalization of the known polynomial time algorithm for trees to a larger class of graphs (note that trees form a proper subclass of block graphs).
\end{abstract}

\section{Introduction}

A \emph{reconfiguration problem} on a graph $G$ looks at how two feasible solutions to a computational problem on $G$ relate to one another (all graphs considered in this paper are simple and undirected unless otherwise mentioned). It involves determining, given a graph $G$ and two feasible solutions of some computational problem on $G$, whether there is a step-by-step transformation from one solution to the other through a series of intermediate solutions adhering to a set of rules. Note that each intermediate solution also has to be a feasible solution to the same computational problem on $G$. Reconfiguration problems have been studied for various computational problems like the independent set problem~\cite{sliding_block,ts_btw,Bonsma16,ITO20111054}, dominating set problem~\cite{ts_bip-per,ts_ds2}, vertex cover problem~\cite{ItoNZ16,MouawadNRS18,ITO20111054} matching problem~\cite{BonamyBHIKMMW19,ITO20111054,SolomonS21} and vertex colouring problem~\cite{sliding_block,ts_vc}.

The ``independent set reconfiguration'' problem can defined as follows.
Suppose that $G$ is a graph and $C,C'$ are two independent sets of $G$. Imagine that each vertex in $C$ has a token placed on it. We would like to determine if there is a sequence of transformation steps that can be applied on the set of tokens so that the tokens eventually are on the vertices of $C'$. Each transformation step must be in accordance with a predetermined set of rules, and at the end of the move, the vertices having tokens on them have to again form an independent set of $G$. Based on the set of rules governing the transformation steps, mainly three types of independent set reconfiguration problems have been studied in the literature --- namely, the ``token sliding''~\cite{BONAMY20216,lokshtanov2018bipartite}, the ``token jumping''~\cite{ito2014fixed,lokshtanov2018bipartite} and the ``token addition/removal''~\cite{lokshtanov2018bipartite} independent set reconfiguration problems.
(Note that these three models have been studied also for reconfiguration problems other than the independent set reconfiguration problem~\cite{ITO202343,ItoNZ16,HADDADAN201637,HaasS14}.)
It has been shown that the token jumping and the token addition/removal models are equivalent for independent set reconfiguration~\cite{ts_cograph}. In the token sliding model introduced by Hearn and Demaine~\cite{sliding_block}, which will be the focus of our study, the only allowed transformation step is the movement of a token from the vertex it is on to one of its neighbours (the token can be imagined to be ``sliding'' along the edge between them in the graph). 

Formally, in the token sliding model for the independent set reconfiguration problem, given a graph $G$ and two independent sets $C,C'$ of $G$ such that $|C|=|C'|$, it has to be determined if there exists a sequence of independent sets $C=C_0, C_1, \ldots, C_k=C'$, where $k\geq 0$, such that for each $i\in\{0,1,\ldots,k-1\}$, $C_{i+1} =(C_i\setminus \{u\}) \cup \{v\}$, for some $u\in C_i$ and edge $uv$ of $G$. Note that if such a sequence exists then $|C_0|=|C_1|=\cdots=|C_k|$, and therefore the problem is trivial if $|C|\neq |C'|$ (in this case, $C$ cannot be transformed into $C'$). So it is customary to assume that the two input independent sets are of the same cardinality. We call this decision problem \prob, which we define formally in Section~\ref{sec:notations}.

The problem \prob\ was first observed to be PSPACE-complete for general graphs by Hearn and Demaine~\cite{sliding_block}. In fact, their result implies that the problem is PSPACE-complete even for subcubic planar graphs (see~\cite{BONSMA20095215,ts_cograph}).
Later, the problem was shown to be PSPACE-complete even for perfect graphs by Kami\'nski, Medvedev and~Milani\v{c}~\cite{ts_cograph}, and this was further improved by Lokshtanov and Mouawad~\cite{lokshtanov2018bipartite}, who showed that the problem remains PSPACE-hard even for bipartite graphs. It was shown that the problem is PSPACE-hard also for another subclass of perfect graphs called split graphs by Belmonte et al.~\cite{belmonte2021tokensplit}. Note that split graphs form a subclass of chordal graphs and even hole-free graphs, and hence the problem is PSPACE-complete for chordal graphs and even-hole free graphs as well. Wrochna~\cite{ts_btw} showed that the problem is PSPACE-complete when restricted to graphs of bounded bandwidth, which implies that the problem is PSPACE-complete for graphs of bounded treewidth, or in fact bounded pathwidth.

Demaine et al.~\cite{DEMAINE2015132tree} showed that that \prob\ is polynomial time solvable for trees. Polynomial time algorithms for the problem were obtained for claw-free graphs by Bonsma, Kami\'nski and Wrochna~\cite{ts_claw-free}, for $P_4$-free graphs by Kami\'nski, Medvedev and Milani\v{c}~\cite{ts_cograph}, for interval graphs by Bonamy and Bousquet~\cite{bonamy}, and for bipartite permutation graphs and bipartite distance-hereditary graphs by Fox-Epstein et al.~\cite{ts_bip-per}.

A \emph{block graph} is an undirected graph in which each biconnected component is a clique.
Block graphs form a subclass of chordal graphs and a superclass of trees. We prove that \prob\ is polynomial-time solvable on block graphs.
It is worth noting that our algorithm is the first generalization of the polynomial time algorithm of Demaine et al. for trees, since none of the other classes for which polynomial time algorithms for the problem have been obtained contains the class of trees. (It was claimed in~\cite{anh2016slidingcactus} and~\cite{hoang2017slidingblock} that the problem can be solved in polynomial time on cactus graphs and block graphs respectively, but the authors of both papers have later announced that the algorithms are incorrect~\cite{noteoncactus,noteonblock} (also see~\cite{hoang2019shortestspider,arxiv}), and that the complexity of the problem on both classes of graphs remains unresolved.)
\medskip

\subsection*{Our contribution}
In this paper, we present a polynomial time algorithm for \prob\ on block graphs. Our algorithm takes as input a block graph $G$ and two independent sets $C_1,C_2$ of $G$ and determines in $O(|E(G)|^4)$ time whether the independent set $C_1$ can be transformed into the independent set $C_2$ using token sliding.
\medskip



Our algorithm is a generalization of the polynomial-time algorithm for trees given in~\cite{DEMAINE2015132tree}. The algorithm for trees first determines which tokens are ``rigid'' in both input independent sets. A token in an independent set is said to be rigid if it cannot be moved from its position by any sequence of valid token movements. Determining which tokens in an independent set are rigid is not very difficult for trees --- a token on a vertex $u$ is rigid if every neighbour $v$ of $u$ has a neighbour $x$ other than $u$ itself that has a token that is rigid for the tree in the forest $T-\{v\}$ containing $x$. This implies that it can be determined in linear time whether a token is rigid. Clearly, two independent sets whose set of vertices containing rigid tokens are unequal cannot be reconfigured with each other. Even if the vertices having rigid tokens is the same for both independent sets, if the two independent sets have an unequal number of tokens in any connected component that is obtained by the removal of the vertices containing the rigid vertices, it can be concluded that the independent sets are not reconfigurable with each other. It is shown in~\cite{DEMAINE2015132tree} that any two independent sets for which the rigid tokens are on exactly the same vertices, and for whom removal of these vertices creates a graph that contains the same number of tokens from each independent set in each connected component can be reconfigured with each other. Further, their proof yields an algorithm that outputs a reconfiguration sequence of length $O(n^2)$ between the two input independent sets. Note that the result above implies that any two independent sets in a tree that have the same cardinality and contain no rigid tokens, are reconfigurable with each other. 

Even though similar to trees at first sight, the problem on block graphs presents many difficulties that hinder the construction of a polynomial-time algorithm similar to the one in~\cite{DEMAINE2015132tree}. For example, unlike for trees, two independent sets in a block graph that have the same cardinality and have no rigid tokens may not be reconfigurable with each other; see Figure~\ref{fig:image1}, or Figure~12 of~\cite{DEMAINE2015132tree}. This hints that instead of rigidity of tokens, one should perhaps find out which tokens are ``trapped'' within the blocks they are in (this is similar to the notion of ``confined'' tokens in~\cite{anh2016slidingcactus}). However, determining which tokens are trapped in their blocks turns out to be not so straightforward.

A more important difference is perhaps the fact that in a block graph, there can be situations in which to insert one additional token into a branch of the block graph, an arbitrary number of tokens may first need to be taken out of that branch --- a situation that never arises for a tree. An example of this phenomenon can be demonstrated as follows. In Figure~\ref{fig:chain}, a block graph with some tokens placed on the vertices of an independent set is shown.
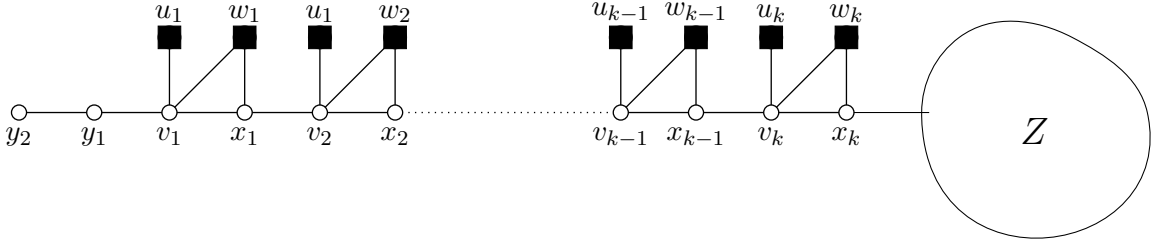
\begin{figure}
\centering
\begin{tikzpicture}
\renewcommand{\defradius}{0.1}
\renewcommand{\vertexset}{(y2,0,0),(y1,1,0),(v1,2,0),(x1,3,0),(v2,4,0),(x2,5,0),(vkm,8,0),(xkm,9,0),(vk,10,0),(xk,11,0),(u1,2,1,black,0.15),(w1,3,1,black,0.15),(u2,4,1,black,0.15),(w2,5,1,black,0.15),(ukm,8,1,black,0.15),(wkm,9,1,black,0.15),(uk,10,1,black,0.15),(wk,11,1,black,0.15)}
\renewcommand{\edgeset}{(y2,y1),(y1,v1),(v1,x1),(x1,v2),(v2,x2),(x2,vkm,,,,dotted),(vkm,xkm),(xkm,vk),(vk,xk),(v1,u1),(v1,w1),(x1,w1),(v2,u2),(v2,w2),(x2,w2),(vkm,ukm),(vkm,wkm),(xkm,wkm),(vk,uk),(vk,wk),(xk,wk)}
\draw (12,0) to [quick curve through={(12,0) (12.5,1) .. (14,1) .. (15, 0).. (12,-0.05)}] (12,0) ;
\draw (11,0) -- (12.1,0);
\drawgraph
\filldraw [fill=black,draw=black] (1.85,0.85) rectangle (2.15,1.15);
\filldraw [fill=black,draw=black] (2.85,0.85) rectangle (3.15,1.15);
\filldraw [fill=black,draw=black] (3.85,0.85) rectangle (4.15,1.15);
\filldraw [fill=black,draw=black] (4.85,0.85) rectangle (5.15,1.15);
\filldraw [fill=black,draw=black] (7.85,0.85) rectangle (8.15,1.15);
\filldraw [fill=black,draw=black] (8.85,0.85) rectangle (9.15,1.15);
\filldraw [fill=black,draw=black] (9.85,0.85) rectangle (10.15,1.15);
\filldraw [fill=black,draw=black] (10.85,0.85) rectangle (11.15,1.15);

\node [below=2] at (\xy{y1}) {$y_1$};
\node [below=2] at (\xy{y2}) {$y_2$};
\node [below=2] at (\xy{v1}) {$v_1$};
\node [below=2] at (\xy{v2}) {$v_2$};
\node [below=2] at (\xy{vkm}) {$v_{k-1}$};
\node [below=2] at (\xy{vk}) {$v_k$};
\node [below=2] at (\xy{x1}) {$x_1$};
\node [below=2] at (\xy{x2}) {$x_2$};
\node [below=2] at (\xy{xkm}) {$x_{k-1}$};
\node [below=2] at (\xy{xk}) {$x_k$};
\node [above=2] at (\xy{u1}) {$u_1$};
\node [above=2] at (\xy{w1}) {$w_1$};
\node [above=2] at (\xy{u2}) {$u_1$};
\node [above=2] at (\xy{w2}) {$w_2$};
\node [above=2] at (\xy{ukm}) {$u_{k-1}$};
\node [above=2] at (\xy{wkm}) {$w_{k-1}$};
\node [above=2] at (\xy{uk}) {$u_k$};
\node [above=2] at (\xy{wk}) {$w_k$};
\node at (13.5,-0.25) {\Large $Z$};
\end{tikzpicture}
\caption{Diagram showing tokens placed in a part of a block graph (tokens represented by black squares)}
\label{fig:chain}
\end{figure}
In the figure, $2k$ of the tokens are shown as black squares. There might be an arbitrary number of tokens in the set $Z$ of vertices. For each $i\in\{1,2,\ldots,k\}$, let $A_i=\{y_1,y_2\}\cup\bigcup_{j\in\{1,2,\ldots,i\}}\{x_i,u_i,v_i,w_i\}$. It can be proved that in order to place a token on $y_2$, we will have to go through an intermediate independent set in which there are at most $k$ tokens in $A_k$ (which means that at least $k$ tokens that were originally in $A_k$ are outside $A_k$ at this point). This can be easily proved by induction on $k$. If $k=1$, then clearly, before reaching an independent set with a token on $y_2$, we will have to reach an independent set with a token on $v_1$, at which point there will be at most one token in $A_1$. Let us assume inductively that the statement is true for $k-1$. Suppose that through a sequence of reconfiguration steps, we place a token on $y_2$. Then we know by the inductive hypothesis that there is some intermediate independent set in which there are at most $k-1$ tokens in $A_{k-1}$. Consider the first such independent set in the sequence of reconfiguration steps. By the choice of this independent set, it is clear that there is a token on $v_k$ in this independent set, which implies that there are no tokens on any other vertex of $A_k\setminus A_{k-1}$. This means that there are at most $k$ tokens in $A_k$ in this independent set, and we are done. The above phenomenon motivates our definition of ``$k$-accessibility'' between independent sets of one branch of a block graph (see Definition~\ref{def:kacc}).

Using the construction described above, one can create situations in which it does not seem straightforward to determine whether a token is ``trapped'' within its block. For example, in the situation shown in Figure~\ref{fig:trapped}, the token on vertex $x$ can be moved out of the block in which it is in, but determining that fact does not seem to be as easy as it is for the case of trees.
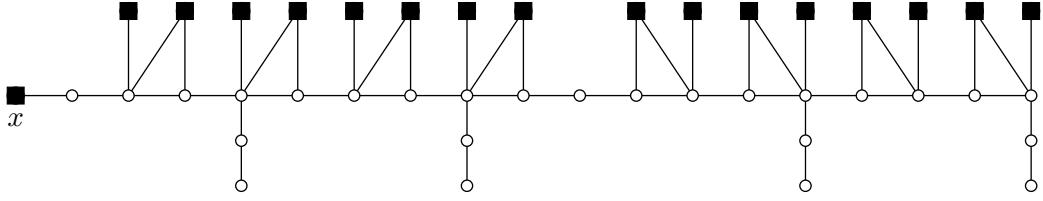
\begin{figure}
\centering
\begin{tikzpicture}[scale=1.5]
\renewcommand{\defradius}{0.05}
\renewcommand{\vertexset}{(v1,1,0,black,0.075),(x1,1.5,0),
(v2,2,0),(x2,2.5,0),(v3,3,0),(x3,3.5,0),(v4,4,0),(x4,4.5,0),(v5,5,0),(x5,5.5,0),
(u2,2,0.75,black,0.075),(w2,2.5,0.75,black,0.075),(u3,3,0.75,black,0.075),(w3,3.5,0.75,black,0.075),(u4,4,0.75,black,0.075),(w4,4.5,0.75,black,0.075),(u5,5,0.75,black,0.075),(w5,5.5,0.75,black,0.075),
(v'2,10,0),(x'2,9.5,0),(v'3,9,0),(x'3,8.5,0),(v'4,8,0),(x'4,7.5,0),(v'5,7,0),(x'5,6.5,0),
(u'2,10,0.75,black,0.075),(w'2,9.5,0.75,black,0.075),(u'3,9,0.75,black,0.075),(w'3,8.5,0.75,black,0.075),(u'4,8,0.75,black,0.075),(w'4,7.5,0.75,black,0.075),(u'5,7,0.75,black,0.075),(w'5,6.5,0.75,black,0.075),(m,6,0),(a5,5,-0.4),(b5,5,-0.8),(a3,3,-0.4),(b3,3,-0.8),
(a'4,8,-0.4),(b'4,8,-0.8),(a'2,10,-0.4),(b'2,10,-0.8)}
\renewcommand{\edgeset}{(v1,x1),(x1,v2),(v2,x2),(x2,v3),(v3,x3),(x3,v4),(v4,x4),(x4,v5),(v5,x5),
(v2,u2),(v2,w2),(x2,w2),(v3,u3),(v3,w3),(x3,w3),(v4,u4),(v4,w4),(x4,w4),(v5,u5),(v5,w5),(x5,w5),
(v'2,x'2),(x'2,v'3),(v'3,x'3),(x'3,v'4),(v'4,x'4),(x'4,v'5),(v'5,x'5),
(v'2,u'2),(v'2,w'2),(x'2,w'2),(v'3,u'3),(v'3,w'3),(x'3,w'3),(v'4,u'4),(v'4,w'4),(x'4,w'4),(v'5,u'5),(v'5,w'5),(x'5,w'5),
(x'5,m),(x5,m),(v5,a5),(a5,b5),(v3,a3),(a3,b3),
(v'4,a'4),(a'4,b'4),(v'2,a'2),(a'2,b'2)}
\drawgraph
\node [below=3] at (\xy{v1}) {$x$};
\filldraw [fill=black,draw=black] (0.925,-0.075) rectangle (1.075,0.075);
\filldraw [fill=black,draw=black] (1.925,0.675) rectangle (2.075,0.825);
\filldraw [fill=black,draw=black] (2.425,0.675) rectangle (2.575,0.825);
\filldraw [fill=black,draw=black] (2.925,0.675) rectangle (3.075,0.825);
\filldraw [fill=black,draw=black] (3.425,0.675) rectangle (3.575,0.825);
\filldraw [fill=black,draw=black] (3.925,0.675) rectangle (4.075,0.825);
\filldraw [fill=black,draw=black] (4.425,0.675) rectangle (4.575,0.825);
\filldraw [fill=black,draw=black] (4.925,0.675) rectangle (5.075,0.825);
\filldraw [fill=black,draw=black] (5.425,0.675) rectangle (5.575,0.825);
\filldraw [fill=black,draw=black] (6.425,0.675) rectangle (6.575,0.825);
\filldraw [fill=black,draw=black] (6.925,0.675) rectangle (7.075,0.825);
\filldraw [fill=black,draw=black] (7.425,0.675) rectangle (7.575,0.825);
\filldraw [fill=black,draw=black] (7.925,0.675) rectangle (8.075,0.825);
\filldraw [fill=black,draw=black] (8.425,0.675) rectangle (8.575,0.825);
\filldraw [fill=black,draw=black] (8.925,0.675) rectangle (9.075,0.825);
\filldraw [fill=black,draw=black] (9.425,0.675) rectangle (9.575,0.825);
\filldraw [fill=black,draw=black] (9.925,0.675) rectangle (10.075,0.825);
\end{tikzpicture}
\caption{The token on $x$ is not trapped within its block}
\label{fig:trapped}
\end{figure}

\section{Notations and preliminaries}\label{sec:notations}
As usual, we let $V(G)$ and $E(G)$ denote the vertex set and edge set of a graph $G$. An edge between vertices $u,v\in V(G)$ is denoted as $uv$. For $S\subseteq V(G)$, we denote by $G-S$ the graph obtained by removing the vertices in $S$ from $G$; i.e. $V(G-S)=V(G)\setminus S$ and $E(G-S)=E(G)\setminus\{uv\in E(G)\colon u\in S\}$. Similarly if $F\subseteq E(G)$, then we denote by $G-F$ the graph obtained by removing the edges in $F$ from $G$; i.e. $V(G-F)=V(G)$ and $E(G-F)=E(G)\setminus F$. For $S\subseteq V(G)$, we let $G[S]=G-(V(G)\setminus S)$. The graph $G[S]$ is commonly known as the subgraph that is induced in $G$ by $S$. A subgraph $H$ of $G$ is said to be an \emph{induced subgraph} of $G$ if $H=G[S]$ for some $S\subseteq V(G)$.
For a graph $G$ and $u\in V(G)$, we denote by $N_G(u)$ the neighbourhood of $u$ in $G$. For $S \subseteq V$, we define $N_G(S)=\bigcup_{v\in S}N_G(v)\setminus S$. When $H$ is a subgraph of $G$, we abbreviate $N_G(V(H))$ to just $N_G(H)$.

\subsubsection*{Reachability}
For two independent sets $I_1,I_2$ of $G$, we say that ``$I_1\rightarrow I_2$'' if one of the following holds:
\begin{itemize}
\vspace{-0.05in}
\itemsep 0in
    \item $I_1=I_2$, or
    \item $\exists v_1,v_2\in V(G)$ such that $I_1\setminus I_2=\{v_1\}$, $I_2\setminus I_1=\{v_2\}$ and $v_1v_2\in E(G)$.
\end{itemize}
Note that the relation $\rightarrow$ is symmetric and reflexive.
For independent sets $I,I'$ of a graph $G$, we say that $I'$ \emph{is reachable from} $I$ in $G$ if there exist independent sets $I_0,I_1,\ldots,I_k$ of $G$, where $k\geq 0$, such that $I=I_0\rightarrow I_1\rightarrow\cdots\rightarrow I_k=I'$. Clearly, the relation ``is reachable from'' is an equivalence relation on the set of independent sets of $G$.
Also, since the relation $\rightarrow$ preserves cardinality, i.e. since $|I_1|=|I_2|$ if $I_1\rightarrow I_2$, we have that two independent sets of $G$ that are reachable from each other are of the same cardinality.
\medskip

The decision problem \prob\ is defined as follows.
\medskip

\begin{center}
\begin{tabular}{|l|}
\hline
\prob\\[0.1in]
\begin{tabular}{ll}
\textbf{Input:}&A graph $G$, independent sets $C_1$, $C_2$ of $G$\\[0.05in]
\textbf{Output:}&``Yes'', if $C_1$ is reachable from $C_2$ in $G$,\\
&``No'', otherwise
\end{tabular}\\
\hline
\end{tabular}
\end{center}
\medskip

The following observations, which are implicitly assumed in the literature about the token sliding independent set reconfiguration problem, are easy to see (see for example Observation~4.1 in~\cite{BelmonteHLOO20}).
\begin{observation}\label{obs:disconnected}
Let $G$ be any graph and $I_1,I_2$ be independent sets of $G$. Then $I_1$ is reachable from $I_2$ if and only if $I_1\cap V(H)$ is reachable from $I_2\cap V(H)$ for every connected component $H$ of $G$.
\end{observation}
\begin{observation}\label{obs:subgraph}
Let $G$ be any graph, $H$ a subgraph of $G$, and $I_1,I_2$ be independent sets of $H$. If $I_1$ is reachable from $I_2$ in $H$, then $I_1$ is reachable from $I_2$ in $G$.
\end{observation}
\subsubsection*{Block graphs}
A \emph{cut-vertex} of a graph $G$ is a vertex $u\in V(G)$ such that $G-\{u\}$ has more connected components than $G$.
A \emph{block} in a graph $G$ is a ``2-connected component of $G$''; i.e. it is a maximal set $S\subseteq V(G)$ such that the subgraph of $G$ induced by $S$ is a 2-connected graph (a connected graph containing no cut-vertices).
A \emph{block graph} is a graph in which each block is a clique.
The following folklore observation can be easily seen to follow from the definition of a block graph.
\begin{observation}\label{obs:induced}
If $G$ is a block graph, then every induced subgraph of $G$ is also a block graph.
\end{observation}

By Observation~\ref{obs:disconnected}, it is enough to solve the \prob\ problem on connected graphs.
Let $G$ be a connected block graph. Note that every vertex of $G$ is contained in some block of $G$. The vertices of $G$ that are contained in more than one block are exactly the cut-vertices of $G$. Every pair of adjacent vertices in $G$ belongs to a unique block of $G$. We denote by $V_{cut}(G)$ the set of cut-vertices of $G$ and by $\mathcal{B}(G)$ the set of blocks of $G$. It is easy to see that if $u\in V(G)\setminus V_{cut}(G)$, then there is a unique block in $\mathcal{B}(G)$ that contains $u$. For any $u\in V(G)$, we denote by $\mathcal{B}_u(G)$ the set of all blocks in $\mathcal{B}(G)$ that contain $u$. For any $B\in\mathcal{B}(G)$, we define $K_B(G)=B\cap V_{cut}(G)$, i.e. $K_B(G)$ is the set of cut-vertices of $G$ that also belong to $B$.

A block graph is commonly represented by its \emph{block-tree} which is a tree having vertex set $V_{cut}(G)\cup\mathcal{B}(G)$ and edge set $\{uB\colon u\in V_{cut}(G), B\in\mathcal{B}_u(G)\}$ (see Figure~\ref{fig:tree}). It is folklore that this definition indeed gives a tree (see for example Section~3.1 of~\cite{Diestel}).
\begin{figure}[t]
    \centering    \includegraphics{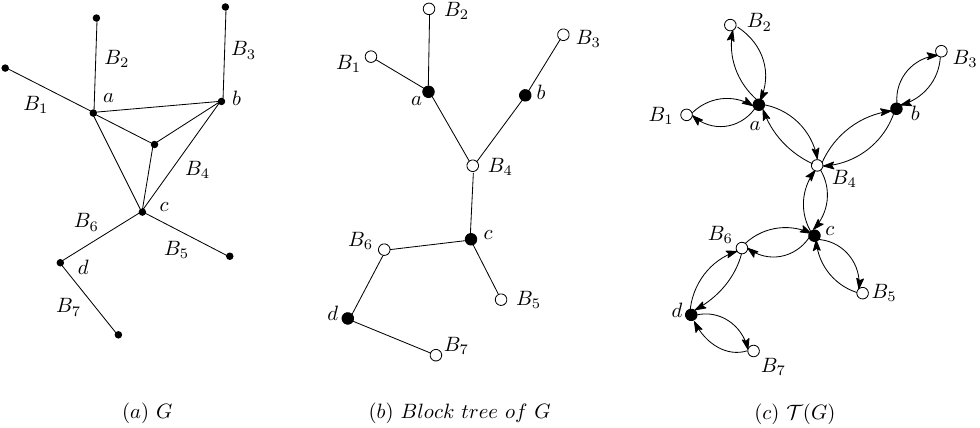}
    \caption{A block graph, its corresponding tree representation and $\mathcal{T}(G)$.}
    \label{fig:tree}
\end{figure}

For a connected block graph $G$, let $\mathcal{T}(G)$ denote its block-tree with each undirected edge replaced with a pair of directed edges, one pointing in each direction. Thus $E(\mathcal{T}(G))=\{(u,B),(B,u)\colon u\in V_{cut}(G), B\in\mathcal{B}_u(G)\}$. For ease of notation we let $\mathcal{P}_G=E(\mathcal{T}(G))$. For each $p\in\mathcal{P}_G$, we now define an induced subgraph $G[p]$ of $G$ as follows.

For $u\in V_{cut}(G)$ and $B\in\mathcal{B}_u(G)$, we define $G[(u,B)]$ (respectively $G[(B,u)]$) as the connected component containing $u$ in the graph $G-\{uv\in E(G)\colon v\in B\}$ (resp. $G-\{uv\in E(G)\colon v\notin B\}$). Clearly, $G[(u,B)]$ and $G[(B,u)]$ are connected graphs. By Observation~\ref{obs:induced}, we have that $G[(u,B)]$ and $G[(B,u)]$ are both block graphs. Most of the time, we abbreviate $G[(u,B)]$ and $G[(B,u)]$ to just $G[u,B]$ and $G[B,u]$ respectively (see Figure~\ref{fig:def} for a schematic diagram of an example block graph $G$ and two induced graphs $G[u,B]$ and $G[B,u]$ of it). We also abbreviate
$\mathcal{B}_u(G[u,B])$ and $K_B(G[B,u])$ to $\beta_G(u,B)$ and $\kappa_G(B,u)$
respectively. We omit the subscript $G$ when the graph being considered is clear from the context.
\begin{figure}[t]
    \centering    \includegraphics{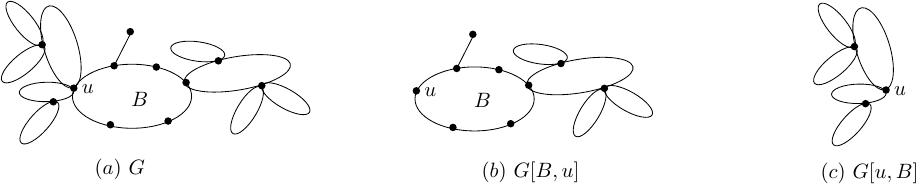}
    \caption{A block graph $G$ and its relevant subgraphs}
    \label{fig:def}
\end{figure}

The following observation is easy to see.
\begin{observation}\label{obs:blockgraph}
Let $G$ be a connected block graph, $H$ a connected subgraph of $G$, and $u\in V_{cut}(G)$. Then either $u\in V(H)$ or there exists $B\in\mathcal{B}_u(G)$ such that $V(H)\subseteq V(G[B,u])\setminus\{u\}$ and for each $B'\in\mathcal{B}_u(G)\setminus B$, $V(H)\cap V(G[B',u])=\emptyset$.
\end{observation}

We shall sometimes write just $\mathcal{P}$ instead of $\mathcal{P}_G$, when the graph $G$ is clear from the context.
Let $G$ be a connected block graph, $u\in V_{cut}(G)$, $B\in\mathcal{B}_u(G)$ and $p\in\{(u,B),(B,u)\}$. We say that $u$ ``is the base of $p$''.
Further, if $p=(u,B)$, then we let $\overline{p}=(B,u)$ and if $p=(B,u)$, then we let $\overline{p}=(u,B)$.

For $p\in\mathcal{P}_G$, a set $C\subseteq V(G[p])$ that is an independent set of $G$ is called a \emph{$p$-independent set} of $G$.
Further, we define $\interior{C}=C\setminus\{u\}$, where $u$ is the base of $p$.
For an independent set $S$ of $G$, we let $S[p]$ denote $V(G[p])\cap S$.
Notice that for an independent set $S$ of $G$, the set $S[p]$ is a $p$-independent set of $G$. Depending on the context, we shall sometimes consider a $p$-independent set of $G$ to be an independent set of $G$ as well.
Given an independent set (or $p$-independent set) $C$ of $G$, we say that a vertex $u$ is \emph{under attack} in $C$ if $N_G(u)\cap C\neq\emptyset$. Note that if $u\in C$, then $u$ is not under attack in $C$.

\section{Some definitions}\label{sec:def}
We first define, for a connected block graph $G$ and for each $p\in\mathcal{P}_G$, a non-negative integer $d_G(p)$ which we shall call the ``depth'' of $p$ in $G$.

\begin{definition}\label{def:depth}
Let $G$ be a connected block graph and $p\in\mathcal{P}_G$.
If $p=(B,u)$, for some $u\in V_{cut}(G)$ and $B\in\mathcal{B}_u(G)$, and $G[B,u]$ is a complete graph (equivalently, $\kappa_G(B,u)=\emptyset$), then we define $d_G(p)=d_G(B,u)=0$. Otherwise, we define:
$$d_G(p)=\left\{\begin{array}{ll}1+\max\{d_G(v,B)\colon v\in\kappa_G(B,u)\}&\mbox{if }p=(B,u)\mbox{ for some }u\in V_{cut}(G),B\in\mathcal{B}_u(G)\\
1+\max\{d_G(B',u)\colon B'\in\beta_G(u,B)\}&\mbox{if }p=(u,B)\mbox{ for some }u\in V_{cut}(G),B\in\mathcal{B}_u(G)\end{array}\right.$$
\end{definition}

Note that as before, we abbreviate $d_G(p)$ to just $d(p)$ when the graph $G$ is clear from the context.
\medskip

Similarly, for each $p\in\mathcal{P}_G$, we define a Boolean value $ua_G(p)$ inductively as follows. (Intuitively, if $ua_G(p)=True$, it roughly indicates that the base $u$ of $p$ is under attack in any $p$-independent set of $G$ in which the tokens cannot be moved away from $u$ so as to create space in $G[p]-\{u\}$ for more tokens to enter via $u$.)

\begin{definition}\label{def:ua}
Let $G$ be a connected block graph and $p\in\mathcal{P}_G$.

If $d_G(p)=0$ then we define $ua_G(p)=True$.

If $p=(B,u)$ for some $u\in V_{cut}(G)$, $B\in\mathcal{B}_u(G)$, and $d_G(p)>0$, we define:
$$ua_G(p)=ua_G(B,u)=\neg\left(\displaystyle\bigwedge_{\forall v\in\kappa_G(B,u)} ua_G(v,B) \land (B=\kappa_G(B,u)\cup\{u\} )\right)$$

On the other hand, if $p=(u,B)$ for some $u\in V_{cut}(G)$, $B\in\mathcal{B}_u(G)$, we define:
$$ua_G(p)=ua_G(u,B)=\displaystyle\bigvee_{\forall B'\in\beta_G(u,B)} ua_G(B',u)$$
\end{definition}

As before, we drop the subscript $G$ from $ua_G(p)$ when the graph under consideration is clear from the context.
\medskip

\noindent\textbf{Note:} While evaluating an arithmetic expression that contains Boolean values, we replace all $True$ values with 1 and all $False$ values with 0.
\medskip

Suppose that $C$ is a $p$-independent set of a connected block graph $G$ for some $p\in\mathcal{P}_G$. We define the ``capacity of $C$'', denoted by $cap(C)$, inductively as follows. (Intuitively, the capacity of a $p$-independent set roughly corresponds to the number of new tokens that can be inserted via the base $u$ of $p$ into $G[p]-\{u\}$ while ensuring that no token in $G[p]-\{u\}$ moves towards $u$.)

\begin{definition}\label{def:cap}
Let $G$ be a connected block graph, $p\in\mathcal{P}_G$ and $C$ a $p$-independent set of $G$.

If $p=(B,u)$ for some $u\in V_{cut}(G)$, $B\in\mathcal{B}_u(G)$, then we define:
$$cap(C)=\sum_{v\in\kappa_G(p)} cap(C[v,B]) + ua_G(p)-|B\cap\interior{C}|.$$

If $p=(u,B)$ for some $u\in V_{cut}(G)$, $B\in\mathcal{B}_u(G)$, we define:
$$cap(C)=\left\{\begin{array}{ll}0&\mbox{if }\exists B'\in\mathcal{B}_u[p]\mbox{ such that }cap(C[B',u])=0\\&\mbox{and }ua_G(B',u)=True\medskip\\\displaystyle\sum_{B'\in\beta_G[p]} \big(cap(C[B',u])-ua_G(B',u)\big) + ua_G(p)&\mbox{otherwise}\end{array}\right.$$
\end{definition}

\renewcommand{\thesubfigure}{(\alph{subfigure})}
\begin{figure}
    \centering
    \begin{subfigure}{0.45\linewidth}
        \centering
        \includegraphics[scale=.75]{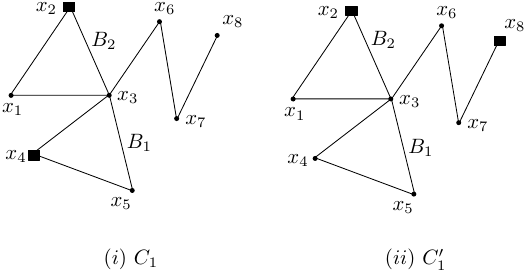}
        \caption{A block graph $G_1$  with two independent sets $C_1$ and $C'_1$ which are not reconfigurable with each other. Here, $cap(C_1[B_1,x_3])=cap(C_1[B_2,x_3])=0$ and $ua_{G_1}(B_1,x_3)=ua_{G_1}(B_2,x_3)=True$.}
        \label{fig:image1}
    \end{subfigure} \hspace{.5cm}
    \begin{subfigure}{0.47\linewidth}
        \centering
        \includegraphics[scale=.75]{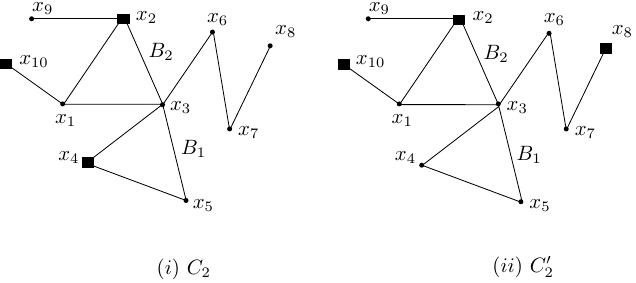}
         \caption{A block graph $G_2$ with two independent sets $C_2$ and $C'_2$ which are reconfigurable with each other.  Here, $cap(C_2[B_1,x_3])=cap(C_2[B_2,x_3])=0$, $ua_{G_2}(B_1,x_3)=True$ and $ua_{G_2}(B_2,x_3)=False$.}
        \label{fig:image2}
    \end{subfigure}
    \caption{Effect of the parameters $cap(C[p])$ and $ua_G(p)$ on reachability}
    \label{fig:notreachable}
\end{figure}

Figure~\ref{fig:notreachable} depicts how the parameters $cap(C[p])$ and $ua_G(p)$, for some $p\in\mathcal{P}_G$, affect the possibility of moving a token to the base of $p$. Consider the graph $G_1$ and the independent sets $C_1, C'_1$ in Figure~\ref{fig:image1}. In this example, the values $cap(C_1[B_2,x_3])=0$ and $ua_{G_1}(B_2,x_3)=True$, together impose some properties (see Lemma~\ref{lem:ua}\ref{it:uatrue}) on $C_1[B_2,x_3]$ that prevent the token on $x_4$ from reaching $x_3$, which makes $C_1$ and $C'_1$ not reconfigurable with each other. Now consider the graph $G_2$ and the independent sets $C_2, C'_2$ in Figure~\ref{fig:image2}. In $G_2$, since $ua_{G_2}(B_2,x_3)=False$, $C_2[B_2,x_3]$ exhibits a different property (see Lemma~\ref{lem:uafalse}) that allows the token on $x_4$ to reach $x_3$, and in turn $x_8$, making $C'_2$ reachable from $C_2$.

\begin{observation}\label{obs:capcomplete}
Let $G$ be a connected block graph, $p\in\mathcal{P}_G$, and $C$ a $p$-independent set of $G$. If $d_G(p)=0$, then $cap(C)=1-|\interior{C}|$.
\end{observation}
\begin{proof}
From Definition~\ref{def:depth}, we know that there exists $u\in V_{cut}(G)$ and $B\in\mathcal{B}_u(G)$ such that $p=(B,u)$ and $G[p]$ is a complete graph. Note that this means that $V(G[p])=B$. It follows directly from Definition~\ref{def:ua} that $ua_G(p)=True$. From Definition~\ref{def:cap} and the fact that $\kappa_G(p)=\emptyset$, we then have that $cap(C)=1-|B\cap\interior{C}|$. Since $V(G[p])=B$, we have $B\cap\interior{C}=\interior{C}$. Thus $cap(C)=1-|\interior{C}|$.
\end{proof}

Let $G$ be a connected block graph.
Let $u\in V_{cut}(G)$ and $\mathcal{A}\subseteq\mathcal{B}_u(G)$. We say that an independent set $C'$ of $G$ is \emph{$(\mathcal{A},u)$-reachable} from an independent set $C$ of $G$ if there exists independent sets $C_0,C_1,\ldots,C_k$ of $G$, where $k\geq 0$, such that $C=C_0\rightarrow C_1\rightarrow\cdots\rightarrow C_k=C'$ and for each $B\in\mathcal{B}_u(G)\setminus\mathcal{A}$, $C_0[B,u]=C_1[B,u]=\cdots=C_k[B,u]$.
Observe that the relation ``is $(\mathcal{A},u)$-reachable from'' is also an equivalence relation defined on the independent sets of $G$.
We shall write ``$(B,u)$-reachable'' as a shorthand for $(\{B\},u)$-reachable and ``$(u,B)$-reachable'' as a shorthand for $(\mathcal{B}_u(G)\setminus\{B\},u)$-reachable. Notice that if $p\in\mathcal{P}_G$, and $C,C'$ are independent sets of $G$ such that $C'$ is $p$-reachable from $C$, then there exist independent sets $C_0,C_1,\ldots,C_k$, where $k\geq 0$, such that $C_0[\overline{p}]=C_1[\overline{p}]=\cdots=C_k[\overline{p}]$.
The following observation can be easily seen to be true.
\begin{observation}\label{obs:reachsubset}
Let $G$ be a connected block graph, $u\in V_{cut}(G)$, $\mathcal{A}\subseteq\mathcal{B}_u(G)$, and $p\in\mathcal{P}_G$.
If $V(G[p])\subseteq\bigcup_{B\in\mathcal{A}} V(G[B,u])$ and $C$ is $p$-reachable from $C'$, then $C$ is $(\mathcal{A},u)$-reachable from $C'$.
\end{observation}

Let $p,p'\in\mathcal{P}_G$ such that $u\in V_{cut}(G)$ is the base of both $p$ and $p'$. A $p$-independent set $Q$ and a $p'$-independent set $Q'$ of $G$ are said to be ``compatible with'' each other if $Q\cap\{u\}=Q'\cap\{u\}$; i.e. $u\in Q$ if and only if $u\in Q'$. Notice that for $p\in\mathcal{P}_G$, if $Q$ is a $p$-independent set, $R$ is a $\overline{p}$-independent set, and $Q$ and $R$ are compatible with each other, then $Q\cup R$ is an independent set of $G$.

\begin{definition}\label{def:kacc}
Let $G$ be a connected block graph, $p\in\mathcal{P}_G$ and $Q$ a $p$-independent set of $G$. For $k\in\mathbb{N}$, we say that a $p$-independent set $Q'$ of $G$ is \emph{$k$-accessible} from $Q$ if $Q'$ is compatible with $Q$ and for every $\overline{p}$-independent set $R$ of $G$ that is compatible with $Q$ (and therefore also with $Q'$) such that $cap(R)\geq k$, $Q\cup R$ is reachable from $Q'\cup R$.
\end{definition}
Notice that for $p\in\mathcal{P}_G$, the relation ``is $k$-accessible from'' defined on the $p$-independent sets of $G$ is an equivalence relation.
\begin{definition}\label{def:strongacc}
Let $G$ be a connected block graph. For $p\in\mathcal{P}_G$, we say that a $p$-independent set $Q'$ of $G$ is \emph{strongly accessible} from a $p$-independent set $Q$ of $G$ if $Q'$ is compatible with $Q$ and for every $\overline{p}$-independent set $R$ of $G$ that is compatible with $Q$ (and therefore also with $Q'$), $Q\cup R$ is $p$-reachable from $Q'\cup R$.
\end{definition}
Again, it is not difficult to see that for $p\in\mathcal{P}_G$, the relation ``is strongly accessible from'' defined on the $p$-independent sets of $G$ is an equivalence relation.
Note that if $Q'$ is strongly accessible from $Q$, then $Q'$ is also $0$-accessible from $Q$.
The following observations follow directly from the definitions and hence are given without proof.
\begin{observation}\label{obs:access}
Let $G$ be a connected block graph and $p\in\mathcal{P}_G$.
\begin{enumerate}
\renewcommand{\theenumi}{(\roman{enumi})}
\renewcommand{\labelenumi}{(\roman{enumi})}
\item If a $p$-independent set $Q$ of $G$ is $k$-accessible from another $p$-independent set $Q'$ of $G$ for some $k\in\mathbb{N}$, then $|Q|=|Q'|$ and $|\overline{Q}|=|\overline{Q'}|$.
\item If a $p$-independent set $Q$ of $G$ is $k$-accessible from another $p$-independent set $Q'$ of $G$ for some $k\in\mathbb{N}$, then $Q$ is $j$-accessible from $Q'$ for each $j\geq k$.
\item Any $p$-independent set $Q$ of $G$ is strongly accessible from itself, and therefore $k$-accessible from itself for any $k\in\mathbb{N}$.
\end{enumerate}
\end{observation}

\section{Token movement: the local view around a cut-vertex}\label{sec:local}

We present some lemmas that will be helpful in the movement of tokens with respect to a given cut-vertex of the graph.
For the whole of this section, we assume that $G$ is a connected block graph. 

\begin{lemma}\label{lem:ua}
Let $p\in\{(u,B),(B,u)\}$, where $u\in V_{cut}(G)$ and $B\in\mathcal{B}_u(G)$, and $C$ be a $p$-independent set of $G$. Then the following statements are true:
\begin{enumerate}
\renewcommand{\theenumi}{(\roman{enumi})}
\renewcommand{\labelenumi}{(\roman{enumi})}
\item\label{it:capnonnegative} $cap(C)\geq 0$.
\item\label{it:uatrue} If $ua(p)=True$ and $u$ is not under attack in $C$, then $cap(C)>0$.
\end{enumerate}
\end{lemma}
\begin{proof}
We prove this by induction on $d(p)$. For the base case, we consider the case when $d(p)=0$. Then we have from Observation~\ref{obs:capcomplete} that $cap(C)=1-|\interior{C}|$. From Definition~\ref{def:depth}, we know that $G[p]$ is a complete graph. This implies that $|\interior{C}|\in\{0,1\}$, from which it follows that $cap(C)\geq 0$. If $u$ is not under attack in $C$, then $|\interior{C}|=0$, and therefore $cap(C)=1-|\interior{C}|=1>0$. For the inductive step, we assume that for any $p'\in\mathcal{P}$ with $d(p')<d(p)$, statements~\ref{it:capnonnegative} and~\ref{it:uatrue} both hold.
\medskip

\noindent\textit{Proof of~\ref{it:capnonnegative}:}

Suppose that $p=(u,B)$ for some $u\in V_{cut}(G)$ and $B\in\mathcal{B}_u(G)$. By Definition~\ref{def:depth}, we know that $d(B',u)<d(p)$ for each $B'\in\beta(p)$. Then by the induction hypothesis on statement~\ref{it:capnonnegative}, we have that $cap(C[B',u])\geq 0$ for each $B'\in\beta(p)$. If there exists $B'\in\beta(p)$ such that $cap(C[B',u])=0$ and $ua(B',u)=True$, then we have by Definition~\ref{def:cap} that $cap(C)=0$, and so we are done. So we can assume that for each $B'\in\beta(p)$, either $cap(C[B',u])>0$ or $ua(B',u)=False$. It then follows that $\sum_{B'\in\beta(p)}(cap(C[B',u])-ua(B',u))\geq 0$, which implies by Definition~\ref{def:cap} that $cap(C)\geq 0$, and we are done. Next, suppose that $p=(B,u)$. We assume for the sake of contradiction that $cap(C)<0$. Again, we have by Definition~\ref{def:depth} that $d(v,B)<d(p)$ for each $v\in\kappa(p)$. By the induction hypothesis on statement~\ref{it:capnonnegative}, we then have that $cap(C[v,B])\geq 0$ for each $v\in\kappa(p)$. Then from Definition~\ref{def:cap} and the assumption that $cap(C)<0$, we have that $ua(p)=False$, $|B\cap\interior{C}|=1$ and $cap(C[v,B])=0$ for each $v\in\kappa(p)$. Let $B\cap\interior{C}=\{v\}$ (clearly $u\neq v$ since $C$ is a $(B,u)$-independent set of $G$ and therefore $u\notin\interior{C}$).
Since $ua(p)=False$, we have from Definition~\ref{def:ua} that $B=\kappa(B,u)\cup\{u\}$ and that $ua(v,B)=True$. The former implies that $v\in\kappa(B,u)$, and then as observed before, we have $cap(C[v,B])=0$. As we have $d(v,B)<d(B,u)$ from Definition~\ref{def:depth}, we conclude using the induction hypothesis on statement~\ref{it:uatrue} that $v$ is under attack in $C[v,B]$. This is a contradiction to the fact that $v\in C$.
\medskip

\noindent\textit{Proof of~\ref{it:uatrue}:}

First, suppose that $p=(B,u)$ for some $u\in V_{cut}(G)$ and $B\in\mathcal{B}_u(G)$. Since $u$ is not under attack in $C$, we have that $|B\cap\interior{C}|=0$. As before, we have $d(v,B)<d(B,u)$ for each $v\in\kappa(p)$, and therefore we have by the induction hypothesis on statement~\ref{it:capnonnegative} that $cap(C[v,B])\geq 0$ for each $v\in\kappa(p)$. Then since $ua(p)=True$, we have from Definition~\ref{def:cap} that $cap(C)>0$, and we are done.
Next, suppose that $p=(u,B)$. 
Note that since $u$ is not under attack in $C$, $u$ is not under attack in $C[B',u]$ for any $B'\in\beta(p)$. For each $B'\in\beta(p)$ such that $ua(B',u)=True$, we have $d(B',u)<d(p)$, and therefore we have by the induction hypothesis on statement~\ref{it:uatrue} that $cap(C[B',u])>0$. Moreover, we have by the induction hypothesis on statement~\ref{it:capnonnegative} that for each $B'\in\beta(p)$, $cap(C[B',u])\geq 0$. This implies that $\sum_{B'\in\beta(p)} (cap(C[B',u])-ua(B',u))\geq 0$.
Now by Definition~\ref{def:cap} and the fact that $ua(p)=True$, we have that $cap(C)>0$.
\end{proof}

\begin{lemma}\label{lem:reconfig}
Let $u\in V_{cut}(G)$ and let $C$ be an independent set of $G$ such that $u\notin C$.
Let $\mathcal{A}\subseteq\mathcal{B}_u(G)$ and for each $B\in\mathcal{A}$, let $Q_B$ be a $(B,u)$-independent set of $G$ that is strongly accessible from $C[B,u]$. Then $\bigcup_{B\in\mathcal{A}} Q_B\cup\bigcup_{B\in\mathcal{B}_u(G)\setminus\mathcal{A}} C[B,u]$ is an independent set of $G$ that is $(\mathcal{A},u)$-reachable from $C$.
\end{lemma}
\begin{proof}
Let $\mathcal{A}=\{B_1,B_2,\ldots,B_s\}$. For $i\in\{1,2\ldots,s\}$, we shall abbreviate $Q_{B_i}$ to just $Q_i$ for ease of notation.
Let $C_0=C$. For $i\in\{1,2,\ldots,s\}$, we inductively define $C_i=C_{i-1}[u,B_i]\cup Q_i$. We claim that for each $i\in\{0,1,\ldots,s\}$, $C_i$ is an independent set of $G$ that is $(\mathcal{A},u)$-reachable from $C$ such that for all $j\in\{1,2,\ldots,i\}$, $C_i[B_j,u]=Q$, for all $j\in\{i+1,i+2,\ldots,s\}$, $C_i[B_j,u]=C[B_j,u]$, and for all $B\in\mathcal{B}_u(G)\setminus\mathcal{A}$, $C_i[B,u]=C[B,u]$. We prove this claim by induction on $i$. Clearly, when $i=0$, we have $C_0=C$, and our claim is trivially true. For the inductive step, we assume that $i>0$ and that $C_{i-1}$ satisfies the requirements of the claim. From the induction hypothesis, we have that $C_{i-1}[B_i,u]=C[B_i,u]$. Since $Q_i$ is strongly accessible from $C[B_i,u]$, we have that $C_i=C_{i-1}[u,B_i]\cup Q_i$ is $(B_i,u)$-reachable from $C_{i-1}=C_{i-1}[u,B_i]\cup C[B_i,u]$, and hence by Observation~\ref{obs:reachsubset}, we have that $C_i$ is $(\mathcal{A},u)$-reachable from $C_{i-1}$. Since $C_{i-1}$ is $(\mathcal{A},u)$-reachable from $C$ by the induction hypothesis, we now have that $C_i$ is $(\mathcal{A},u)$-reachable from $C$.
It is clear from the definition of $C_i$
that for all $j\in\{1,2,\ldots,i-1\}$, we have $C_i[B_j,u]=C_{i-1}[B_j,u]=Q_j$, that $C_i[B_i,u]=Q_i$, and that for each $j\in\{i+1,i+2,\ldots,s\}$, we have $C_i[B_j,u]=C_{i-1}[B_j,u]=C[B_j,u]$. Clearly, for any $B\in\mathcal{B}_u(G)\setminus\mathcal{A}$, $C_i[B,u]=C_{i-1}[B,u]=C[B,u]$. This proves the claim. Then $C_s$ is an independent set of $G$ that is $(\mathcal{A},u)$-reachable from $C$ such that for each $i\in\{1,2,\ldots,s\}$, $C_s[B_i,u]=Q_i$ and for each $B\in\mathcal{B}_u(G)\setminus\mathcal{A}$, $C_s[B,u]=C[B,u]$. The proof is completed by noting that $C_s=\bigcup_{B\in\mathcal{A}} Q_B\cup\bigcup_{B\in\mathcal{B}_u(G)\setminus\mathcal{A}} C[B,u]$.
\end{proof}

\begin{lemma}\label{lem:uafalse}
Let $u\in V_{cut}(G)$, $B\in\mathcal{B}_u(G)$ and $p\in\{(B,u),(u,B)\}$.
Let $Q$ be a $p$-independent set of $G$.
If $ua(p)=False$ or $cap(Q)>0$, then there exists a $p$-independent set $Q'$ of $G$ that is strongly accessible from $Q$ such that $cap(Q')\geq cap(Q)$ and $u$ is not under attack in $Q'$.
\end{lemma}
\begin{proof}
If $u$ is not under attack in $Q$, then the statement of the lemma is trivially true as we can just set $Q'=Q$.
So we assume that $u$ is under attack in $Q$, which implies that $u\notin Q$.

We use induction on $d(p)$. As the base case, we consider the case when $d(p)=0$. Then by Observation~\ref{obs:capcomplete}, we have that $cap(Q)=1-|\interior{Q}|$. As $u$ is under attack in $Q$, and $G[p]$ is a complete graph by Definition~\ref{def:depth}, we have that $|\interior{Q}|=1$. Thus $cap(Q)=1-|\interior{Q}|=0$. Since $ua(p)=True$ by Definition~\ref{def:ua}, we can conclude that the statement of the lemma is vacuously true. For the inductive step, we assume that $d(p)>0$ . We shall assume inductively that the statement of the lemma holds for any $p'\in\mathcal{P}$ such that $d(p')<d(p)$.

Suppose that $p=(u,B)$. Let $\beta(u,B)=\{B_1,B_2,\ldots,B_s\}$. If there exists $i\in\{1,2,\ldots,s\}$ such that $ua(B_i,u)=True$ and $cap(Q[B_i,u])=0$, then by Definition~\ref{def:cap}, we have that $cap(Q)=0$, and by Definition~\ref{def:ua}, we have $ua(p)=True$. Then the statement of the lemma is vacuously true, and we are done. So using Lemma~\ref{lem:ua}\ref{it:capnonnegative}, we can assume that for each $i\in\{1,2,\ldots,s\}$, we have either $ua(B_i,u)=False$ or $cap(Q[B_i,u])>0$.
Then by the induction hypothesis, we know that for each $i\in\{1,2,\ldots,s\}$, there exists a $(B_i,u)$-independent set $Q_i$ of $G$ that is strongly accessible from $Q[B_i,u]$ such that $u$ is not under attack in $Q_i$ and $cap(Q_i)\geq cap(Q[B_i,u])$.

Consider any $(B,u)$-independent set $R$ of $G$ that is compatible with $Q$; i.e. $u\notin R$. Let $D=Q\cup R$. Clearly, for each $i\in\{1,2,\ldots,s\}$, we have that $D[B_i,u]=Q[B_i,u]$, which means that $Q_i$ is a $(B_i,u)$-independent set that is strongly accessible from $D[B_i,u]$ such that $cap(Q_i)\geq cap(D[B_i,u])$.
By applying Lemma~\ref{lem:reconfig} (setting $C=D$, $\mathcal{A}=\beta(u,B)=\{B_1,B_2,\ldots,B_s\}$, $Q_{B_i}=Q_i$ for each $i\in\{1,2,\ldots,s\}$), we have that $D'=Q_1\cup Q_2\cup\cdots\cup Q_s\cup R$ is an independent set of $G$ that is $(u,B)$-reachable from $D$.
Thus for any $(B,u)$-independent set $R$ of $G$ that is compatible with $Q$, we have that $D'[u,B]\cup R$ is $(u,B)$-reachable from $D[u,B]\cup R=Q\cup R$.
This means that $D'[u,B]$ is strongly accessible from $D[u,B]=Q$. Further, $D'[u,B]=Q_1\cup Q_2\cup\cdots\cup Q_s$, and therefore $u$ is not under attack in $D'[u,B]$. As $cap(D'[B_i,u])=cap(Q_i)\geq cap(D[B_i,u])$ for each $i\in\{1,2,\ldots,s\}$, we have from Definition~\ref{def:cap} that $cap(D'[u,B])\geq cap(D[u,B])=cap(Q)$. Thus $Q'=D'[u,B]$ is a $(u,B)$-independent set that is strongly accessible from $Q$ such that $u$ is not under attack in $Q'$ and $cap(Q')\geq cap(Q)$. This completes the proof for the case $p=(u,B)$.

Next, suppose $p=(B,u)$. Suppose that either $cap(Q)>0$ or $ua(B,u)=False$.
Let $R$ be a $(u,B)$-independent set of $G$ that is compatible with $Q$ (equivalently, $u\notin R$) and let $D=Q\cup R$. Notice that since $u$ is under attack in $Q$, we have that $B\cap\interior{Q}\neq\emptyset$.
If $cap(Q)>0$, then we have by Definition~\ref{def:cap} that there exists $v\in\kappa(B,u)$ such that $cap(Q[v,B])>0$. Similarly, if $ua(B,u)=False$, then Definition~\ref{def:cap} again implies that there exists $v\in\kappa(B,u)$ such that $cap(Q[v,B])>0$. So we have that in any case, there exists $v\in\kappa(B,u)$ such that $cap(Q[v,B])>0$.
Since $d(v,B)<d(B,u)$, we can apply the induction hypothesis to conclude that there exists a $(v,B)$-independent set $S$ that is strongly accessible from $Q[v,B]$ such that $v$ is not under attack in $S$ and $cap(S)\geq cap(Q[v,B])=cap(D[v,B])$. 
This means that there exists an independent set $D_1=S\cup D[B,v]$ of $G$ (notice that $S$ is compatible with $Q[v,B]=D[v,B]$ and hence also compatible with $D[B,v]$) such that $D_1$ is $(v,B)$-reachable from $D$. As $V(G[v,B])\subseteq V(G[B,u])$, we now have from Observation~\ref{obs:reachsubset} that $D_1$ is $(B,u)$-reachable from $D$.
Then $D_1$ is an independent set of $G$ that is $(B,u)$-reachable from $D$ such that $v$ is not under attack in $D_1[v,B]$ and it follows from Definition~\ref{def:cap} that $cap(D_1[B,u])\geq cap(D[B,u])=cap(Q)$. Observe
that $D_1[B,u]\cap B=Q\cap B$. Thus $u$ is under attack in $D_1[B,u]$, since it is under attack in $Q$.

Let $D_2=(D_1\setminus B)\cup\{v\}$. Clearly, we have $D_1\rightarrow D_2$, and further that $D_2$ is $(B,u)$-reachable from $D_1$, and hence also from $D$. Observe that for each $v'\in\kappa(B,u)$, $\interior{D_1[v',B]}=\interior{D_2[v',B]}$, and therefore by Definition~\ref{def:cap}, we have $cap(D_1[v',B])=cap(D_2[v',B])$. This implies, again by Definition~\ref{def:cap}, that $cap(D_2[B,u])=cap(D_1[B,u])\geq cap(Q)$.
Note that we have $cap(D_2[v,B])=cap(D_1[v,B])=cap(S)\geq cap(Q[v,B])>0$.
By Definition~\ref{def:cap}, this means that there does not exist $B'\in\beta(v,B)$ such that $cap(D_2[B',v])=0$ and $ua(B',v)=True$. Since $cap(D_2[v,B])>0$, this implies by Definition~\ref{def:cap} and Definition~\ref{def:ua} that there exists $B'\in\beta(v,B)$ such that $cap(D_2[B',v])>0$. Now we choose $B^\star\in\beta(v,B)$ as follows. We choose as $B^\star$ a block in $\beta(v,B)$ such that $cap(D_2[B^\star,v])>1\vee (cap(D_2[B^\star,v])=1\wedge ua(B^\star,v)=False)$, if such a block exists. Otherwise, we choose as $B^\star$ a block in $\beta(v,B)$ such that $cap(D_2[B^\star,v])=1$. Note that $B^\star$ is well defined. Recall that $v\in D_2$.
We shall show the existence of an independent set $D'$ of $G$ that is $(B,u)$-reachable from $D$ such that
$u$ is not under attack in $D'[B,u]$, and $cap(D'[B,u])\geq cap(Q)$.

Let us first assume that for all $v'\in\kappa(B^\star,v)$, we have $ua(v',B^\star)=True$ and $cap(D_2[v',B^\star])=0$.
As $cap(D_2[B^\star,v])>0$, we have by Definition~\ref{def:cap} that $ua(B^\star,v)=True$. Then we have by Definition~\ref{def:ua} that there exists $w\in B^\star\setminus(\kappa(B^\star,v)\cup\{v\})$. We now define $D'=(D_2\setminus\{v\})\cup\{w\}$. Note that $D'$ is an independent set of $G$ and $D_2\rightarrow D'$, which further implies that $D'$ is $(B,u)$-reachable from $D_2$, and hence also from $D$. Moreover, $u$ is not under attack in $D'[B,u]$. It can be seen using Definition~\ref{def:cap} that $cap(D'[v,B])=cap(D_2[v,B])-1$ and further that $cap(D'[B,u])=cap(D_2[B,u])\geq cap(Q)$.

Next, we consider the case when there exists $v'\in\kappa(B^\star,v)$, such that either $ua(v',B^\star)=False$ or $cap(D_2[v',B^\star])>0$.
As $v'\in\kappa(B^\star,v)$, $B^\star\in\beta(v,B)$, and $v\in\kappa(B,u)$, we have $d(v',B^\star)<d(B^\star,v)<d(v,B)<d(B,u)$.
By the induction hypothesis, there exists a $(v',B^\star)$-independent set $S'$ of $G$ that is strongly accessible from $D_2[v',B^\star]$ such that $v'$ is not under attack in $S'$ and $cap(S')\geq cap(D_2[v',B^\star])$.
Then there exists an independent set $D_3=S'\cup D_2[B^\star,v']$ of $G$ that is $(v',B^\star)$-reachable from $D_2$. Since $V(G[v',B^\star])\subseteq V(G[B,u])$, it follows from Observation~\ref{obs:reachsubset} that $D_3$ is $(B,u)$-reachable from $D_2$, and hence also from $D$. Since $D_3[v',B^\star]=S'$, we have that $v'$ is not under attack in $D_3[v',B^\star]$. Note that $cap(D_3[v',B^\star])=cap(S')\geq cap(D_2[v',B^\star])$, which implies by Definition~\ref{def:cap} that $cap(D_3[B^\star,v])\geq cap(D_2[B^\star,v])$ and further that $cap(D_3[B,u])\geq cap(D_2[B,u])\geq cap(Q)$.
We now define $D'=(D_3\setminus\{v\})\cup\{v'\}$. It is clear that $D'$ is an independent set of $G$ such that $D_3\rightarrow D'$, which also means that $D'$ is $(B,u)$-reachable from $D_3$, and hence also from $D$. Also, it is easy to see that $u$ is not under attack in $D'[B,u]$. As before, it can be seen using Definition~\ref{def:cap} that $cap(D'[v,B])=cap(D_3[v,B])-1$ and that $cap(D'[B,u])=cap(D_3[B,u])\geq cap(Q)$.

Thus in both cases, for any choice of $R$, we have an independent set $D'=D'[B,u]\cup R$ of $G$ that is $(B,u)$-reachable from $D=D[B,u]\cup R=Q\cup R$ such that $u$ is not under attack in $D'[B,u]$ and $cap(D'[B,u])\geq cap(Q)$.
So $Q'=D'[B,u]$ is a $(B,u)$-independent set of $G$ that is strongly accessible from $Q$ such that $u$ is not under attack in $Q'$ and $cap(Q')\geq cap(Q)$. This completes the proof.
\end{proof}

\begin{lemma}\label{lem:inserttoken}
Let $u\in V_{cut}(G)$ and $B\in\mathcal{B}_u(G)$. Let $D$ be an independent set of $G$ such that $u \in D$. If $cap(D[B,u])>0$, then there exists an independent set $D'$ of $G$ such that $D'$ is reachable from $D$, $u\notin D'$, $\interior{D'[u,B]}=\interior{D[u,B]}$ and $cap(D'[B,u])\geq cap(D[B,u])-1$.
\end{lemma}
\begin{proof}
Since $u\in D$, we know that $u$ is not under attack in $D[B,u]$.

Let us first consider the case when for each $v\in\kappa(B,u)$, we have $ua(v,B)=True$ and $cap(D[v,B])=0$. Then as $cap(D[B,u])>0$, it follows from Definition~\ref{def:cap} that $ua(B,u)=True$, which further implies by Definition~\ref{def:ua} that there exists a vertex $v \in B\setminus (\kappa(B,u)\cup\{u\})$.
We define $D'=(D\setminus \{u\})\cup\{v\}$. Note that $D'$ is an independent set of $G$ and that $D\rightarrow D'$, which means that $D'$ is reachable from $D$. Moreover, we have $\interior{D'[u,B]}=\interior{D[u,B]}$ and $u\notin D'$. It is easy to see using Definition~\ref{def:cap} that $cap(D'[B,u])=cap(D[B,u])-1$. So we are done in this case.

Next we consider the case when there exists $v\in\kappa(B,u)$ such that either $ua(v,B)=False$ or $cap(D[v,B])>0$. We now have by Lemma~\ref{lem:uafalse} that there exists a $(v,B)$-independent set $S$ of $G$ that is strongly-accessible from $D[v,B]$ such that $v$ is not under attack in $S$ and $cap(S)\geq cap(D[v,B])$. Thus $D_1=S\cup D[B,v]$ is an independent set of $G$ that is reachable from $D$. Since $D_1[v,B]=S$, we have that $cap(D_1[v,B])=cap(S)\geq cap(D[v,B])$. It follows from Definition~\ref{def:cap} that $cap(D_1[B,u])\geq cap(D[B,u])$. Further it is clear that $\interior{D_1[u,B]}=\interior{D[u,B]}$. We now define $D'=(D_1\setminus\{u\})\cup\{v\}$. It is easy to see that $D'$ is an independent set of $G$ (recall that $v$ is not under attack in $D_1[v,B]=S$) such that $D_1\rightarrow D'$, which implies that $D'$ is reachable from $D$, $u\notin D'$, and that $\interior{D'[u,B]}=\interior{D_1[u,B]}=\interior{D[u,B]}$. Finally, it can be seen using Definition~\ref{def:cap} that $cap(D'[B,u])=cap(D_1[B,u])-1\geq cap(D[B,u])-1$.
\end{proof}

\begin{lemma}\label{lemma:blockshift}
Let $u\in V_{cut}(G)$ and $B\in\mathcal{B}_u(G)$. Let $D$ be an independent set of $G$ such that $u\notin D$.
If there does not exist distinct $B',B''\in\mathcal{B}_u(G)$ such that $cap(D[B',u])=cap(D[B'',u])=0$ and $ua(B',u)=ua(B'',u)=True$, then there exists a $(u,B)$-independent set $Q$ of $G$ that is $\gamma$-accessible from $D[u,B]$ such that $cap(Q)=\sum\limits_{\forall B'\in\beta(u,B)}\Big(cap(D[B', u])-ua(B', u)\Big)+ua(u,B)$, where $\gamma=\min\{cap(D[B,u]),1\}$.
\end{lemma}
\begin{proof}
If $cap(D[u,B])=\sum\limits_{\forall B'\in\beta(u,B)}\Big(cap(D[B', u])-ua(B', u)\Big)+ua(u,B)$, then we can simply let $Q=D[u,B]$, and we are done (recall Observation~\ref{obs:access}). So let us assume that that is not the case.

Suppose first that $cap(D[B,u])=0$ and $ua(B,u)=True$. From the assumption in the statement of the lemma, we know that there does not exist distinct $B',B''\in\mathcal{B}_u(G)$ such that $cap(D[B',u])=cap(D[B'',u])=0$ and $ua(B',u)=ua(B'',u)=True$. This means that for each $B'\in\beta(u,B)$, either $cap(D[B',u])>0$ or $ua(B',u)=False$. Then we have from Definition~\ref{def:cap} that $cap(D[u,B])=\sum\limits_{\forall B'\in\beta(u,B)}\Big(cap(D[B', u])-ua(B', u)\Big)+ua(u,B)$, which contradicts our assumption above. So we assume that either $cap(D[B,u])>0$ or $ua(B,u)=False$.

As $cap(D[u,B])\neq\sum\limits_{\forall B'\in\beta(u,B)}\Big(cap(D[B', u])-ua(B', u)\Big)+ua(u,B)$, we have from Definition~\ref{def:cap} that there exists $X\in\beta(u,B)$ such that $cap(D[X,u])=0$ and $ua(X,u)=True$, and also that $cap(D[u,B])=0$. Note that from the assumption in the statement of the lemma, this means that for each $B'\in\beta(u,X)$, we have either $cap(D[B',u])>0$ or $ua(B',u)=False$. It now follows from Definition~\ref{def:ua} that $\sum\limits_{\forall B'\in\beta(u,B)}\Big(cap(D[B', u])-ua(B', u)\Big)+ua(u,B)\geq 0$. As $cap(D[u,B])=0$ and $cap(D[u,B])\neq\sum\limits_{\forall B'\in\beta(u,B)}\Big(cap(D[B', u])-ua(B', u)\Big)+ua(u,B)$, we can conclude that $\sum\limits_{\forall B'\in\beta(u,B)}\Big(cap(D[B', u])-ua(B', u)\Big)+ua(u,B)>0$. Then it follows from Definition~\ref{def:cap} that there exists $B^\star\in\beta(u,B)$ such that either $cap(D[B^\star,u])>1\vee (cap(D[B^\star,u])=1\wedge ua(B^\star,u)=False)$. By Lemma~\ref{lem:ua}\ref{it:uatrue} we know that $u$ is under attack in $D[X,u]$.

Now consider any $(B,u)$-independent set $R$ of $G$ such that $u\notin R$ and $cap(R)=\gamma$.
Let $D'=D[u,B]\cup R$. Note that we have either $cap(D'[B,u])=cap(R)=\gamma>0$ or $ua(B,u)=False$ (as $\gamma=0$ implies that $cap(D[B,u])=0$), and that $D'[B',u]=D[B',u]$ for all $B'\in\beta(u,B)$. Thus we can conclude that for each $B'\in\beta(u,X)$, either $cap(D'[B',u])>0$ or $ua(B',u)=False$.
Let $\beta(u,X)=\{B_1,B_2,\ldots,B_s\}$. Note that $B^\star\in\{B_1,B_2,\ldots,B_s\}$. We assume without loss of generality that $B_s=B^\star$. For each $i\in\{1,2,\ldots,s\}$, since we know that either $cap(D'[B_i,u])>0$ or $ua(B_i,u)=False$, we can use Lemma~\ref{lem:uafalse} to conclude that there exists a $(B_i,u)$-independent set $Q_i$ of $G$ that is strongly accessible from $D'[B_i,u]$ such that $u$ is not under attack in $Q_i$ and $cap(Q_i)\geq cap(D'[B_i,u])$. By applying Lemma~\ref{lem:reconfig} (setting $C=D'$, $\mathcal{A}=\beta(u,X)=\{B_1,B_2,\ldots,B_s\}$, and $Q_{B_i}=Q_i$ for each $i\in\{1,2,\ldots,s\}$), we know that $D_1=Q_1\cup Q_2\cup\cdots\cup Q_s\cup D'[X,u]$ is an independent set of $G$ that is reachable from $D'$. Clearly, we have $u\notin D_1$, $u$ is not under attack in $D_1[u,X]$, and that $cap(D_1[B^\star,u])=cap(Q_s)\geq cap(D'[B_s,u])=cap(D'[B^\star,u])$. Define $D_2=(D_1\setminus X)\cup\{u\}$. Notice that $D_1\rightarrow D_2$ (recall that $u$ is under attack in $D_1[X,u]=D'[X,u]$), and therefore $D_2$ is reachable from $D'$. As $\interior{D_2[B^\star,u]}=\interior{D_1[B^\star,u]}$, we have by Definition~\ref{def:cap} that $cap(D_2[B^\star,u])=cap(D_1[B^\star,u])\geq cap(D'[B^\star,u])>0$. Moreover, we have from Definition~\ref{def:cap} that $cap(D_2[X,u])=cap(D_1[X,u])+1=cap(D'[X,u])+1=cap(D[X,u])+1$ ($=1$). Now using Lemma~\ref{lem:inserttoken}, we have that there exists an independent set $D_3$ that is reachable from $D_2$, and hence also from $D'$, such that $u\notin D_3$, $\interior{D_3[u,B^\star]}=\interior{D_2[u,B^\star]}$, and $cap(D_3[B^\star,u])\geq cap(D_2[B^\star,u])-1=cap(D_1[B^\star,u])-1\geq cap(D'[B^\star,u])-1=cap(D[B^\star,u])-1$. Notice that for each $i\in\{1,2,\ldots,s-1\}$, $D[B_i,u]$ is strongly accessible from $D_3[B_i,u]=D_2[B_i,u]=D_1[B_i,u]=Q_i$. Applying Lemma~\ref{lem:reconfig} (setting $C=D_3$, $\mathcal{A}=\{B_1,B_2,\ldots,B_{s-1}\}$, and $Q_{B_i}=D[B_i,u]$ for each $i\in\{1,2,\ldots,s-1\}$), we get that $D_4=D[B_1,u]\cup D[B_2,u]\cup\cdots\cup D[B_{s-1},u]\cup D_3[B_s,u]\cup D_3[X,u]$ is an independent set of $G$ that is reachable from $D_3$, and hence also from $D'$. Define $Q=D_4[u,B]$. Notice that for each $B'\in\mathcal{B}_u(G)\setminus\{X,B^\star\}=\{B_1,B_2,\ldots,B_{s-1}\}$, we have that $D_4[B',u]=D'[B',u]$, which implies that either $cap(D_4[B',u])>0$ or $ua(B',u)=False$. Note that $u\notin D_4$. Since $B\notin\{X,B^\star\}$, we have $D_4[B,u]=R$. Thus $Q=D_4[u,B]$ is $\gamma$-accessible from $D'[u,B]=D[u,B]$. Since $D_4[X,u]=D_3[X,u]=D_2[X,u]$, we have that $cap(D_4[X,u])=cap(D[X,u])+1=1$. Similarly, as $D_4[B^\star,u]=D_3[B^\star,u]$, we also have $cap(D_4[B^\star,u])=cap(D_3[B^\star,u])\geq cap(D[B^\star,u])-1$. Since we know that if $cap(D[B^\star,u])=1$, then $ua(B^\star,u)=False$, we can conclude that either $cap(D_4[B^\star,u])>0$ or $ua(B^\star,u)=False$. Thus we can conclude that for each $B'\in\beta(u,B)$, either $cap(D_4[B',u])>0$ or $ua(B',u)=False$. Now we have from Definition~\ref{def:cap} that
\begin{eqnarray*}
cap(Q)=cap(D_4[u,B])&=&\sum\limits_{\forall B'\in\beta(u,B)}\Big(cap(D_4[B', u])-ua(B',u)\Big)+ua(u,B)\\
&=&\sum\limits_{\forall B'\in\beta(u,B)\setminus\{X,B^\star\}}\Big(cap(D_4[B', u])-ua(B',u)\Big)+ua(u,B)\\&&+cap(D_4[X,u])-ua(X,u)+cap(D_4[B^\star,u])-ua(B^\star,u)\\
&\geq&\sum\limits_{\forall B'\in\beta(u,B)\setminus\{X,B^\star\}}\Big(cap(D[B', u])-ua(B',u)\Big)+ua(u,B)\\&& + cap(D[X,u])+1-ua(X,u)+cap(D[B^\star,u])-1-ua(B^\star,u)\\
&=&\sum\limits_{\forall B'\in\beta(u,B)}\Big(cap(D[B', u])-ua(B',u)\Big)+ua(u,B)
\end{eqnarray*}
This completes the proof.
\end{proof}

The following observation is easy to verify.
\begin{observation}\label{obs:sum}
Let $u\in V_{cut}(G)$ and $B\in\mathcal{B}_u(G)$. For any independent set $C$ of $G$, we have $$\sum\limits_{\forall u'\in \kappa(B,u)} |\interior{C[u',B]}|=|\interior{C[B,u]}|-|B\cap\interior{C[B,u]}|$$ and, $$\sum\limits_{\forall B'\in \beta(u,B)} |\interior{C[B',u]}|=|\interior{C[u,B]}|$$
\end{observation}
\begin{lemma}\label{lem:ureachable}
Let $C_0,C_1,\ldots,C_k$, where $k\geq 0$, be independent sets of $G$ such that $C_0\rightarrow C_1\rightarrow\cdots\rightarrow C_k$, and $u\in V_{cut}(G)$. If $u\notin C_0\cup C_1\cup\cdots\cup C_k$, then
for each $B\in\mathcal{B}_u(G)$, we have that $C_k[B,u]$ is strongly accessible from $C_0[B,u]$.
\end{lemma}
\begin{proof}
Suppose that $C_{i-1}\neq C_i$ for some $i\in\{1,2,\ldots,k\}$. Since $C_{i-1}\rightarrow C_i$, we have that there exist $v_1,v_2\in V(G)$ such that $(C_{i-1}\setminus C_i)=\{v_1\}$, $(C_i\setminus C_{i-1})=\{v_2\}$ and $v_1v_2\in E(G)$. Since $u\notin C_0\cup C_1\cup\cdots\cup C_k$, we have that $u\notin\{v_1,v_2\}$, which implies that $G[\{v_1,v_2\}]$ is a connected subgraph of $G$ that does not contain $u$.
Now by Observation~\ref{obs:blockgraph}, we have that 
there exists $B'\in\mathcal{B}_u(G)$ such that
$v_1,v_2\in V(G[B',u])\setminus\{u\}$. Thus, we can conclude that if $C_{i-1}\neq C_i$ for some $i\in\{1,2,\ldots,k\}$, then there exists $B'\in\mathcal{B}_u(G)$ such that $(C_{i-1}\setminus C_i)\cup (C_i\setminus C_{i-1})\subseteq V(G[B',u])\setminus\{u\}$. In particular, we have $C_{i-1}[B',u]\setminus C_i[B',u]=\{v_1\}$, $C_i[B',u]\setminus C_{i-1}[B',u]=\{v_2\}$, and $C_{i-1}[u,B']=C_i[u,B']$.

Let $B\in\mathcal{B}_u(G)$. Let $R$ be any $(u,B)$-independent set of $G$ that is compatible with $C_0[B,u]$ (i.e. $u\notin R$). In order to prove the statement of the lemma, we define, for each $i\in\{0,1,\ldots,k\}$, an independent set $D_i$ such that $D_i[u,B]=R$, satisfying $C_0[B,u]\cup R=D_0\rightarrow D_1\rightarrow \cdots \rightarrow D_k= C_k[B,u]\cup R$. Let $D_0=C_0[B,u]\cup R$. For each $i\in\{1,2,\ldots,k\}$, we define $D_i$ as follows. If $C_i=C_{i-1}$, then we simply let $D_i=D_{i-1}$. Suppose that $C_i\neq C_{i-1}$. Then we know that there exists $B'\in\mathcal{B}_u(G)$ such that $(C_{i-1}\setminus C_i)\cup (C_i\setminus C_{i-1})\subseteq V(G[B',u])\setminus\{u\}$. 
If $B'=B$, then we let $D_i=C_i[B,u]\cup R$ (since $u\notin C_0\cup C_1\cup\cdots\cup C_k$, $D_i$ is an independent set of $G$).
Otherwise, we have $C_i[B,u]=C_{i-1}[B,u]$, and we let $D_i=D_{i-1}$. Therefore, for each $i\in\{1,2,\ldots,k\}$, we have that $D_i[B,u]=C_i[B,u]$ and $D_i[u,B]=D_{i-1}[u,B]=\cdots=D_0[u,B]=R$. This implies that $D_k=D_k[B,u]\cup D_k[u,B]=C_k[B,u]\cup R$. Now, it only remains to be shown that $D_{i-1}\rightarrow D_i$ for each $i\in\{1,2,\ldots,k\}$. If $C_i=C_{i-1}$, then $D_i=D_{i-1}$, and therefore we have $D_{i-1}\rightarrow D_i$. So suppose that $C_i\neq C_{i-1}$. Then there exists $B'\in\mathcal{B}_u(G)$ such that $(C_{i-1}\setminus C_i)\cup (C_i\setminus C_{i-1})\subseteq V(G[B',u])\setminus\{u\}$.
If $B'\neq B$, then we have $D_{i-1}=D_i$, and therefore $D_{i-1}\rightarrow D_i$, and we are again done.
So we can assume that $B'=B$. Recall that we have $D_i[B,u]=C_i[B,u]$, $D_{i-1}[B,u]=C_{i-1}[B,u]$, and $D_i[u,B]=D_{i-1}[u,B]$. Recall that $C_{i-1}[B,u]\setminus C_i[B,u]=\{v_1\}$, $C_i[B,u]\setminus C_{i-1}[B,u]=\{v_2\}$, where $C_{i-1}\setminus C_i=\{v_1\}$, $C_i\setminus C_{i-1}=\{v_2\}$, and $v_1v_2\in E(G)$. We now have that $D_{i-1}[B,u]\setminus D_i[B,u]=\{v_1\}$, $D_i[B,u]\setminus D_{i-1}[B,u]=\{v_2\}$, and $D_i[u,B]=D_{i-1}[u,B]$. Thus, $D_{i-1}\setminus D_i=\{v_1\}$ and $D_i\setminus D_{i-1}=\{v_2\}$. Since $v_1v_2\in E(G)$, we now have that $D_{i-1}\rightarrow D_i$. This implies that $D_0\rightarrow D_1\rightarrow \cdots \rightarrow D_k$, and we are done.
\end{proof}
\begin{corollary}\label{cor:ureachable}
Let $C_0,C_1,\ldots,C_k$, where $k\geq 0$, be independent sets of $G$ such that $C_0\rightarrow C_1\rightarrow\cdots\rightarrow C_k$, and $u\in V_{cut}(G)$. If $u\notin C_0\cup C_1\cup\cdots\cup C_k$, then:
\begin{enumerate}
\renewcommand{\theenumi}{(\roman{enumi})}
\renewcommand{\labelenumi}{(\roman{enumi})}
\item\label{it:cardinality} for each $B\in\mathcal{B}_u(G)$, we have $|\interior{C_0[B,u]}|=|\interior{C_k[B,u]}|$ and $|\interior{C_0[u,B]}|=|\interior{C_k[u,B]}|$, and
\item\label{it:blocks} 
for any $\mathcal{A}\subseteq\mathcal{B}_u(G)$, $\bigcup_{B\in\mathcal{A}} C_k[B,u]\cup\bigcup_{B\notin\mathcal{A}} C_0[B,u]$ is $(\mathcal{A},u)$-reachable from $C_0$.
\end{enumerate}
\end{corollary}
\begin{proof}
We prove the two statements separately.
\medskip

\noindent\textit{Proof of~\ref{it:cardinality}}:
Let $B\in\mathcal{B}_u(G)$. Since we have by Lemma~\ref{lem:ureachable} that $C_k[B,u]$ is strongly accessible from $C_0[B,u]$, we have by Observation~\ref{obs:access} that $|C_k[B,u]|=|C_0[B,u]|$. As $u\notin C_0\cup C_1\cup\cdots\cup C_k$, we now have $|\interior{C_0[B,u]}|=|C_0[B,u]|=|C_k[B,u]|=|\interior{C_k[B,u]}|$. Since $|C_0|=|C_k|$, it now follows that $|\interior{C_0[u,B]}|=|C_0|-|\interior{C_0[B,u]}|=|C_k|-|\interior{C_k[B,u]}|=|\interior{C_k[u,B]}|$.
\medskip

\noindent\textit{Proof of~\ref{it:blocks}:}
Let $\mathcal{A}\subseteq\mathcal{B}_u(G)$. For each $B\in\mathcal{A}$, define $Q_B=C_k[B,u]$. It follows from Lemma~\ref{lem:ureachable} that for each $B\in\mathcal{A}$, $Q_B$ is strongly accessible from $C_0[B,u]$. We now have from Lemma~\ref{lem:reconfig} that $\bigcup_{B\in\mathcal{A}} C_k[B,u]\cup\bigcup_{B\notin\mathcal{A}} C_0[B,u]$ is $(\mathcal{A},u)$-reachable from $C_0$.
\end{proof}
\section{Potentials: the global perspective}\label{sec:pot}
\begin{definition}\label{def:pot}
For a connected block graph $G$, and $p\in\mathcal{P}_G$, we define the \emph{potential} of $p$ with respect to an independent set $C$ of $G$, denoted by $pot_G(C,p)=\max\{cap(C'[p])+|\interior{C'[p]}|-|\interior{C[p]}|\colon C'$ is an independent set of $G$ that is reachable from $C\}$.    
\end{definition}
When the graph $G$ is clear from the context, we sometimes abbreviate $pot_G(C,p)$ to just $pot(C,p)$.
We claim that the procedure {\sc Compute-potentials}, listed as Algorithm~\ref{alg:cap}, when given a connected block graph $G$ and an independent set $C$ of it, computes the value of $pot_G(C,p)$ for each $p\in\mathcal{P}_G$. 

\begin{algorithm}[t]
  \SetAlgoLined
  \DontPrintSemicolon
  \SetNoFillComment
  \SetCommentSty{}
  \SetKwProg{myproc}{Procedure}{}{}
  \SetKwInOut{Input}{Input}
  \SetKwInOut{Output}{Output}
  
  \myproc{\sc Compute-potentials ($G,C$)}{
    \Input{A block graph $G$ and an independent set $C$ of $G$}
    \Output{$y[p]\in\mathbb{N}$, $\forall p\in\mathcal{P}_G$}
    \For{$p\in\mathcal{P}_G$}{
        $y[p]\gets 0$\\
    }
    $updated\gets True$\\
    \While{$updated$}{
        \label{line:while}
        $updated\gets False$\\
        \For{$p\in\mathcal{P}_G$}{\label{line:for}
            $y'\gets y[p]$\\
            \If{$p=(B,u)$, for some $u\in V_{cut}(G)$ and $B\in\mathcal{B}_u(G)$}{
                $y'\gets\sum\limits_{\forall u'\in \kappa_G(B,u)}y[u',B]+ua_G(B,u)-|B \cap \interior{C[B,u]}|$\label{line:y'Bu}\\
            } 
            \ElseIf{$p=(u,B)$, for some $u\in V_{cut}(G)$, $B\in\mathcal{B}_u(G)$, and $\nexists B', B''\in\mathcal{B}_u(G)$ such that $y[B',u]=y[B'',u]=0$ and $ua_G(B',u)=ua_G(B'',u)=True$}{
                $y'\gets\sum\limits_{\forall B'\in\beta_G(u,B)}\Big(y[B', u]-ua_G(B', u)\Big)+ua_G(u,B)$\label{line:y'uB}\\
            }
            \If{$y[p]<y'$}{
                $y[p]\gets y'$\label{line:assign}\\
            	$updated\gets True$\label{line:updated}\\
                \textbf{break}\label{line:break}\\
            }
        }
    }
    \Return $y$
 }
\caption{The algorithm for computing potentials}
\label{alg:cap}
\end{algorithm}
\medskip

For the rest of this section, we assume that $G$ is a connected block graph and $C$ is an independent set of $G$. For $p\in\mathcal{P}$, let $\mathbf{x}[p]$ denote the final value of the variable $y[p]$ that is computed by the procedure {\sc Compute-potentials}$(G,C)$.
Our aim in this section will be to prove that for each $p\in\mathcal{P}$, $\mathbf{x}[p]=pot(C,p)$.
\medskip

Let us analyze Algorithm~\ref{alg:cap} in detail.
Suppose that the \textbf{while} loop starting on line~\ref{line:while} gets executed $t$ times in total during the execution of Algorithm~\ref{alg:cap}. For each $i\in\{1,2,\ldots,t\}$ and $p\in\mathcal{P}$, we let $y^{(i)}[p]$ denote the value of the variable $y[p]$ just before the $i$-th iteration of the \textbf{while} loop starts. Since it is clear that no variable $y[p]$, for $p\in\mathcal{P}$, is updated during the last iteration of the \textbf{while} loop, it follows from the definition of $\mathbf{x}[p]$ that for each $p\in\mathcal{P}$, $y^{(t)}[p]=\mathbf{x}[p]$.

It is easy to see that the algorithm never decreases the value of a variable $y[p]$, for any $p\in\mathcal{P}$, and that exactly one of the variables $y[p]$, for $p\in\mathcal{P}$, gets updated during every iteration of the \textbf{while} loop other than the last iteration. Thus we have the following two observations.

\begin{observation}\label{obs:ymonotone}
Let $p\in\mathcal{P}$. For $1\leq i<j\leq t$, $y^{(i)}[p]\leq y^{(j)}[p]$, and therefore, $\mathbf{x}[p]=y^{(t)}[p]\geq y^{(1)}[p]=0$.
\end{observation}
\begin{observation}\label{obs:uniquep}
For every $i\in\{2,3,\ldots,t\}$, there exists $p'\in\mathcal{P}$ such that $y^{(i)}[p']>y^{(i-1)}[p']$ and $y^{(i)}[p]=y^{(i-1)}[p]$ for every $p\in \mathcal{P}\setminus\{p'\}$.
\end{observation}

Consider the last iteration of the \textbf{while} loop. Since this is the last iteration of the \textbf{while} loop, line~\ref{line:updated} is not executed during this iteration. This means that the \textbf{break} statement of line~\ref{line:break} is never executed during this iteration, which further implies that the \textbf{for} loop of line~\ref{line:for} gets executed for every $p\in\mathcal{P}$.

\begin{lemma}\label{lem:xnonnegative}
Let $u\in V_{cut}(G)$ and $B\in\mathcal{B}_u(G)$. Then
$$\sum\limits_{\forall u'\in\kappa(B,u)} \mathbf{x}[u',B]+ua(B,u)-|B \cap \interior{C[B,u]}|\geq 0$$
\end{lemma}
\begin{proof}
Suppose for the sake of contradiction that for some $u\in V_{cut}(G)$ and $B\in\mathcal{B}_u(G)$, we have $\sum\limits_{\forall u'\in\kappa(B,u)} \mathbf{x}[u',B]+ua(B,u)-|B \cap \interior{C[B,u]}|<0$. Then
$\sum\limits_{\forall u'\in\kappa(B,u)} y^{(t)}[u',B]+ua(B,u)-|B \cap \interior{C[B,u]}|<0$. By Observation~\ref{obs:ymonotone}, we have that $y^{(t)}[u',B]\geq 0$ for each $u'\in\kappa(B,u)$, which implies that $|B\cap \interior{C[B,u]}|=1$, $ua(B,u)=False$, and $y^{(t)}[u',B]=0$ for each $u'\in\kappa(B,u)$. Let $B\cap \interior{C[B,u]}=\{v\}$. By Definition~\ref{def:ua}, we now have that $v\in\kappa(B,u)$ and $ua(v,B)=True$. Note also that $y^{(t)}[v,B]=0$. Let $\mathcal{B}_v^T=\{B'\in\mathcal{B}_v(G)\colon ua(B',v)=True\}$ (note that it is possible that $B\in\mathcal{B}_v^T$). Let $B'\in\mathcal{B}_v^T$. As $v\in C$, we have that $B'\cap \interior{C[B',v]}=\emptyset$.
Since the \textbf{for} loop gets executed for each $p\in\mathcal{P}$ during the last iteration of the \textbf{while} loop, we have that line~\ref{line:y'Bu} gets executed for $p=(B',v)$ during the last iteration of the \textbf{while} loop. The value of $y'$ that gets computed at this step is greater than 0 since $ua(B',v)=True$, $B'\cap \interior{C[B',v]}=\emptyset$, and $y^{(t)}[u',B']\geq 0$ for all $u'\in\kappa(B',v)$ by Observation~\ref{obs:ymonotone}. Since line~\ref{line:updated} is not executed subsequently, this means that the value of $y[B',v]$ at this point of time, which is $y^{(t)}[B',v]$, is greater than the value of $y'$, and is therefore greater than 0. Thus, we have that for each $B'\in\mathcal{B}_v^T$, $y^{(t)}[B',v]>0$.
This means that line~\ref{line:y'uB} gets executed for $p=(v,B)$ during the last iteration of the \textbf{while} loop. Since $ua(v,B)=True$, $y^{(t)}[B',v]\geq 0$ for each $B'\in\beta(v,B)$ (by Observation~\ref{obs:ymonotone}), and $y^{(t)}[B',v]-ua(B',v)\geq 0$ for each $B'\in\mathcal{B}_v^T$, we have that the value of the variable $y'$ computed during this execution of line~\ref{line:y'uB} is positive. As $y^{(t)}[v,B]=0$, line~\ref{line:updated} would have subsequently been executed, which contradicts the fact that this is the last iteration of the \textbf{while} loop.
\end{proof}

\begin{lemma}\label{lem:xBu}
For any $u\in V_{cut}(G)$ and $B\in\mathcal{B}_u(G)$, $$\mathbf{x}[B,u]=\sum\limits_{\forall u'\in\kappa(B,u)} \mathbf{x}[u',B]+ua(B,u)-|B \cap \interior{C[B,u]}|$$
\end{lemma}
\begin{proof}

Let $u\in V_{cut}(G)$ and $B\in\mathcal{B}_u(G)$. 
Note that line~\ref{line:updated} is not executed during the last iteration of the \textbf{while} loop, and that the \textbf{for} loop is executed for every $p\in\mathcal{P}$ in the last iteration. It follows that the value of $y'$ that is computed in line~\ref{line:y'Bu} when $p=(B,u)$, during the last iteration of the \textbf{while} loop, is at most the value of the variable $y[B,u]$ at that point, which is equal to $y^{(t)}[B,u]$. So we get $y^{(t)}[B,u]\geq\sum\limits_{\forall u'\in\kappa(B,u)} y^{(t)}[u',B]+ua(B,u)-|B \cap \interior{C[B,u]}|$. We thus have $\mathbf{x}[B,u]\geq\sum\limits_{\forall u'\in\kappa(B,u)} \mathbf{x}[u',B]+ua(B,u)-|B \cap \interior{C[B,u]}|$.
Let $i$ be the smallest integer in $\{1,2,\ldots,t\}$ such that $y^{(i)}[B,u]=\mathbf{x}[B,u]$. Clearly, $i$ exists since $y^{(t)}[B,u]=\mathbf{x}[B,u]$. Suppose that $i>1$. By our choice of $i$, we have $y^{(i-1)}[B,u]<y^{(i)}[B,u]$, which implies that the variable $y[B,u]$ is assigned the value $\mathbf{x}[B,u]$ during an execution of line~\ref{line:assign} of the $(i-1)$-th iteration of the \textbf{while} loop. It can be seen that the value of $y'$ at that point of time is equal to $\sum\limits_{\forall u'\in\kappa(B,u)} y^{(i-1)}[u',B]+ua(B,u)-|B \cap \interior{C[B,u]}|$, which is computed during the previous execution of line~\ref{line:y'Bu}, and it is this value that is assigned to $y[B,u]$.
We thus get

\begin{eqnarray*}
\mathbf{x}[B,u]=y^{(i)}[B,u]&=&\sum\limits_{\forall u'\in\kappa(B,u)} y^{(i-1)}[u',B]+ua(B,u)-|B \cap \interior{C[B,u]}|\\
&\leq&\sum\limits_{\forall u'\in\kappa(B,u)} y^{(t)}[u',B]+ua(B,u)-|B \cap \interior{C[B,u]}|\mbox{ (by Observation~\ref{obs:ymonotone})}\\
&=&\sum\limits_{\forall u'\in\kappa(B,u)} \mathbf{x}[u',B]+ua(B,u)-|B \cap \interior{C[B,u]}|
\end{eqnarray*}
and we are done. So we can assume that $i=1$. Then we have $\mathbf{x}[B,u]=y^{(1)}[B,u]=0$. We now have from Lemma~\ref{lem:xnonnegative} that $\mathbf{x}[B,u]=0\leq\sum\limits_{\forall u'\in\kappa(B,u)} \mathbf{x}[u',B]+ua(B,u)-|B \cap \interior{C[B,u]}|$, and we are done.
\end{proof}
\begin{lemma}\label{lem:xuB}
For any $u\in V_{cut}(G)$ and $B\in\mathcal{B}_u(G)$,\\
if $\exists B',B''\in\mathcal{B}_u(G)$ such that $\mathbf{x}[B',u]=\mathbf{x}[B'',u]=0$ and $ua(B',u)=ua(B'',u)=True$, then
$$\mathbf{x}[u,B]=0$$ and otherwise, $$\mathbf{x}[u,B]=\sum\limits_{\forall B'\in\beta(u,B)}\Big(\mathbf{x}[B', u]-ua(B', u)\Big)+ua(u,B)$$
\end{lemma}
\begin{proof}
The argument is similar to that in the previous proof.
Let $u\in V_{cut}(G)$ and $B\in\mathcal{B}_u(G)$.
Suppose that $\exists B',B''\in\mathcal{B}_u(G)$ such that $\mathbf{x}[B',u]=\mathbf{x}[B'',u]=0$ and $ua(B',u)=ua(B'',u)=True$. Then since
$y^{(t)}[p]=\mathbf{x}[p]$ for all $p\in\mathcal{P}$, we have that $y^{(t)}[B',u]=y^{(t)}[B'',u]=0$. By Observation~\ref{obs:ymonotone}, we now have that $y^{(i)}[B',u]=y^{(i)}[B'',u]=0$ for every $i\in\{1,2,\ldots,t\}$. This means that line~\ref{line:y'uB} does not get executed for the pair $(u,B)$ in any iteration of the \textbf{while} loop, which means that $\mathbf{x}[u,B]=y^{(t)}[u,B]=y^{(1)}[u,B]=0$. So we are done in this case. We shall now assume that $\nexists B',B''\in\mathcal{B}_u(G)$ such that $\mathbf{x}[B',u]=\mathbf{x}[B'',u]=0$ and $ua(B',u)=ua(B'',u)=True$; this means that $\nexists B',B''\in\mathcal{B}_u(G)$ such that $y^{(t)}[B',u]=y^{(t)}[B'',u]=0$ and $ua(B',u)=ua(B'',u)=True$. Then line~\ref{line:y'uB} gets executed for the pair $(u,B)$ during the last iteration of the \textbf{while} loop. By Observation~\ref{obs:ymonotone}, we have that if $y^{(t)}[B', u]-ua(B', u)<0$ for some $B'\in\beta(u,B)$, then $y^{(t)}[B',u]=0$ and $ua(B',u)=True$, which means that $y^{(t)}[B', u]-ua(B', u)=-1$. The fact that line~\ref{line:y'uB} got executed for the pair $(u,B)$ during the last iteration of the \textbf{while} loop then implies that there exists at most one block in $\beta(u,B)$, say $B''$, such that $y^{(t)}[B'', u]-ua(B'', u)<0$. If such a $B''$ exists, then $ua(u,B)=True$ (from Definition~\ref{def:ua} and the fact that $ua(B'',u)=True$). It follows that the value of the variable $y'$ computed during the execution of line~\ref{line:y'uB} for the pair $(u,B)$ during the last iteration of the \textbf{while} loop, which is equal to $\sum\limits_{\forall B'\in\beta(u,B)}\Big(y^{(t)}[B', u]-ua(B', u)\Big)+ua(u,B)$, is always non-negative.
Since line~\ref{line:updated} does not get executed during the last iteration of the \textbf{while} loop, it must be the case that this value is at most $y^{(t)}[u,B]$, and therefore we have $y^{(t)}[u,B]\geq \sum\limits_{\forall B'\in\beta(u,B)}\Big(y^{(t)}[B', u]-ua(B', u)\Big)+ua(u,B)$.
If $y^{(t)}[u,B]=\sum\limits_{\forall B'\in\beta(u,B)}\Big(y^{(t)}[B', u]-ua(B', u)\Big)+ua(u,B)$, we are done since $\mathbf{x}[p]=y^{(t)}[p]$ for all $p\in\mathcal{P}$. So we can assume that $y^{(t)}[u,B]>\sum\limits_{\forall B'\in\beta(u,B)}\Big(y^{(t)}[B', u]-ua(B', u)\Big)+ua(u,B)$. By our observations above, we then have that $y^{(t)}[u,B]>0$. Let $i$ be the smallest integer in $\{1,2,\ldots,t\}$ such that $y^{(t)}[u,B]=y^{(i)}[u,B]$. As $y^{(t)}[u,B]>0$, we have that $i>1$. From our choice of $i$, we have that lines~\ref{line:y'uB} and~\ref{line:assign} get executed for the pair $(u,B)$ during the $(i-1)$-th iteration of the \textbf{while} loop. Thus $y^{(t)}[u,B]=y^{(i)}[u,B]=\sum\limits_{\forall B'\in\beta(u,B)}\Big(y^{(i-1)}[B', u]-ua(B', u)\Big)+ua(u,B)$. Now by Observation~\ref{obs:ymonotone}, it follows that $y^{(t)}[u,B]\leq\sum\limits_{\forall B'\in\beta(u,B)}\Big(y^{(t)}[B', u]-ua(B', u)\Big)+ua(u,B)$, which is a contradiction.
\end{proof}

\subsection{Time complexity of the algorithm}

\begin{lemma}\label{lem:xupperbound}
For any $p\in\mathcal{P}$, $\mathbf{x}[p]\leq |\mathcal{B}(G[p])|$.
\end{lemma}
\begin{proof}
We prove this by induction on $d(p)$. For the base case, observe that if $d(p)=0$, then $p=(B,u)$ for some $u\in V_{cut}(G)$ and $B\in\mathcal{B}_u(G)$, $ua(B,u)=True$, and $G[B,u]$ is a complete graph. Then we have by Lemma~\ref{lem:xBu} that $\mathbf{x}[p]=1-|B\cap\overline{C[B,u]}|\leq 1=|\mathcal{B}(G[B,u])|$. This proves the base case. For the inductive step, let us assume that for all $p'\in\mathcal{P}$ such that $d(p')<d(p)$, the statement of the lemma is true. Suppose that $p=(B,u)$ for some $u\in V_{cut}(G)$ and $B\in\mathcal{B}_u(G)$. Then from Lemma~\ref{lem:xBu}, we have that $\mathbf{x}[B,u]=\sum\limits_{\forall u'\in\kappa(B,u)} \mathbf{x}[u',B]+ua(B,u)-|B \cap \interior{C[B,u]}|$. Since $d(u',B)<d(B,u)$ for all $u'\in\kappa(B,u)$, we can use the induction hypothesis to conclude that $\mathbf{x}[B,u]\leq\sum\limits_{\forall u'\in\kappa(B,u)} |\mathcal{B}(G[u',B])|+ua(B,u)-|B \cap \interior{C[B,u]}|$. It is easy to see that $|\mathcal{B}(G[B,u])|=1+\sum\limits_{\forall u'\in\kappa(B,u)} |\mathcal{B}(G[u',B])|$. We thus have that $\mathbf{x}[B,u]\leq |\mathcal{B}(G[B,u])|-1+ua(B,u)-|B \cap \interior{C[B,u]}|\leq |\mathcal{B}(G[B,u])|$.
Next, suppose that $p=(u,B)$ for some $u\in V_{cut}(G)$ and $B\in\mathcal{B}_u(G)$. Then we have by Lemma~\ref{lem:xuB} that either $\mathbf{x}[u,B]=0$ or $\mathbf{x}[u,B]=\sum\limits_{\forall B'\in\beta(u,B)}\Big(\mathbf{x}[B', u]-ua(B', u)\Big)+ua(u,B)$. In the former case, we are done since $|\mathcal{B}(G[u,B])|\geq 0$. So we assume that $\mathbf{x}[u,B]=\sum\limits_{\forall B'\in\beta(u,B)}\Big(\mathbf{x}[B', u]-ua(B', u)\Big)+ua(u,B)$. Since $d(B',u)<d(u,B)$ for all $B'\in\beta(u,B)$, we have by the induction hypothesis that $\mathbf{x}[u,B]\leq\sum\limits_{\forall B'\in\beta(u,B)}\Big(|\mathcal{B}(G[B', u])|-ua(B', u)\Big)+ua(u,B)$. It can be seen from Definition~\ref{def:ua} that if $ua(B',u)=True$ for any $B'\in\beta(u,B)$, then $ua(u,B)=True$. Since $\sum\limits_{\forall B'\in\beta(u,B)} |\mathcal{B}(G[B', u])|=|\mathcal{B}(G[u,B])|$, it follows that $\sum\limits_{\forall B'\in\beta(u,B)}\Big(|\mathcal{B}(G[B', u])|-ua(B', u)\Big)+ua(u,B)\leq |\mathcal{B}(G[u,B])|$. Thus, we get $\mathbf{x}[u,B]\leq |\mathcal{B}(G[u,B])|$, and we are done.
\end{proof}

\begin{lemma}\label{lem:polytime}
The procedure {\sc Compute-potentials}$(G,C)$ runs in time $O(|E(G)|^4)$.
\end{lemma}
\begin{proof}
Let $n=|V_{cut}(G)|$ and $m=|\mathcal{B}(G)|$.
Note that an adjacency list representation of the block tree of $\mathcal{T}(G)$ can be computed in $O(|V(G)|+|E(G)|)$ time~\cite{hoptarj}. Recall that $\mathcal{P}=E(\mathcal{T}(G))$. Thus the operations like determining $\beta(u,B)$ and $\kappa(B,u)$, for some $u\in V_{cut}(G)$ and $B\in\mathcal{B}_u(G)$, become just adjacency checking operations on $\mathcal{T}(G)$.
From Observations~\ref{obs:uniquep} and~\ref{obs:ymonotone}, it follows that the \textbf{while} loop of Algorithm~\ref{alg:cap} gets executed at most $\sum_{p\in\mathcal{P}} \mathbf{x}[p]$ times, which by Lemma~\ref{lem:xupperbound} is at most $|\mathcal{P}|m$ times. It is not difficult to see that $|\mathcal{P}|=2|E(\mathcal{T}(G))|=2(n+m-1)$. Thus, we have that the \textbf{while} loop gets executed at most $2m(n+m-1)$ times. As the \textbf{for} loop gets executed at most $|\mathcal{P}|$ times inside each iteration of the \textbf{while} loop, and it takes $O(n+m)$ time for one iteration of the \textbf{for} loop, we can conclude that each iteration of the \textbf{while} loop takes time at most $O(|\mathcal{P}|(n+m))=O((n+m)^2)$. Thus the total running time of the algorithm is $O(m(n+m)^3)=O(m^4)$ (since $|\mathcal{B}(G)|=m\geq n=|V_{cut}(G)|$). As $m\leq |E(G)|$, we conclude that the algorithm runs in $O(|E(G)|^4)$ time.
\end{proof}
\subsection{Proof of correctness}
\begin{lemma}\label{lem:evolution}
Let $p\in\mathcal{P}$.
For every $i\in\{1,2,\ldots,t\}$, there exists a $p$-independent set $Q$ of $G$ that is $y^{(i)}[\overline{p}]$-accessible from $C[p]$ and has $cap(Q)\geq y^{(i)}[p]$.
\end{lemma}
\begin{proof}
We prove this by induction on $d(p)$.
As the base case, let us consider the situation when $d(p)=0$. Then we have by Definition~\ref{def:depth} that $p=(B,u)$ for some $u\in V_{cut}(G)$ and $B\in\mathcal{B}_u(G)$, and $\kappa(p)=\emptyset$. 
Further, we have by Definition~\ref{def:ua} that $ua(B,u)=True$. Also, by Observation~\ref{obs:capcomplete}, $cap(C[p])=1-|\interior{C[p]}|$. As $\kappa(B,u)=\emptyset$, whenever line~\ref{line:y'Bu} of Algorithm~\ref{alg:cap} is executed for $p=(B,u)$, the value of $y'$ that is computed is equal to $ua(p)-|\interior{C[p]}|=1-|\interior{C[p]}|=cap(C[p])$. It follows that for all $i\in\{1,2,\ldots,t\}$, $y^{(i)}[p]\leq cap(C[p])$.
We can therefore let $Q=C[p]$ and we are done (recall Observation~\ref{obs:access}). For the inductive step, we assume that:

\begin{enumerate}
\item[($*$)]\label{lemmahypothesis}
for all $p'\in\mathcal{P}$ such that $d(p')<d(p)$, and for each $j\in\{1,2,\ldots,t\}$, there exists a $p'$-independent set $Z$ of $G$ such that $Z$ is $y^{(j)}[\overline{p'}]$-accessible from $C[p']$ and $cap(Z)\geq y^{(j)}[p']$
\end{enumerate}

\begin{claim}\label{clm:uBacc}
Suppose that $p=(u,B)$ for some $u\in V_{cut}(G)$ and $B\in\mathcal{B}_u(G)$.
For all $i\in\{1,2,\ldots,t\}$, there exists a $(u,B)$-independent set $Q_i$ of $G$ that is $y^{(i)}[B,u]$-accessible from $C[u,B]$ such that for all $B'\in\beta(u,B)$, $cap(Q_i[B',u])\geq y^{(i)}[B',u]$, and $Q_i[B',u]$ is $y^{(i)}[u,B']$-accessible from $C[B',u]$.
\end{claim}
We prove this claim by induction on $i$. For $i=1$, we simply let $Q_1=C[u,B]$. Clearly, as $y^{(1)}[B,u]=0$, we have by Observation~\ref{obs:access} that $Q_1$ is $y^{(1)}[B,u]$-accessible from $C[u,B]$. For each $B'\in\beta(u,B)$, we have $y^{(1)}[B',u]=0$, $y^{(1)}[u,B']=0$, and $Q_1[B',u]=C[B',u]$, which implies by Observation~\ref{obs:access} that $Q_1[B',u]$ is $y^{(1)}[u,B']$-accessible from $C[B',u]$, and implies by Lemma~\ref{lem:ua}\ref{it:capnonnegative} that $cap(Q_1[B',u])\geq 0=y^{(1)}[B',u]$. For the inductive step, we assume that $i>1$ and that there exists a $(u,B)$-independent set $Q_{i-1}$ of $G$ that is $y^{(i-1)}[B,u]$-accessible from $C[u,B]$ such that for all $B'\in\beta(u,B)$, $Q_{i-1}[B',u]$ is $y^{(i-1)}[u,B']$-accessible from $C[B',u]$ and $cap(Q_{i-1}[B',u])\geq y^{(i-1)}[B',u]$.

If $y^{(i)}[B',u]=y^{(i-1)}[B',u]$ for all $B'\in\beta(u,B)$, then we simply let $Q_i=Q_{i-1}$, and we are done by the induction hypothesis (recall that $y^{(i)}[B,u]\geq y^{(i-1)}[B,u]$ and $y^{(i)}[u,B']\geq y^{(i-1)}[u,B']$ for all $B'\in\beta(u,B)$ from Observation~\ref{obs:ymonotone}). So let us assume that there exists $X\in\beta(u,B)$ such that $y^{(i)}[X,u]\neq y^{(i-1)}[X,u]$, which implies by Observation~\ref{obs:uniquep} that for all $B'\in\beta(u,B)\setminus\{X\}$, $y^{(i)}[B',u]=y^{(i-1)}[B',u]$. Notice that for each $B'\in\beta(u,B)\setminus\{X\}$, $cap(Q_{i-1}[B',u])\geq y^{(i-1)}[B',u]=y^{(i)}[B',u]$.

Since $d(X,u)<d(u,B)$, we have by~($*$) that there exists an $(X,u)$-independent set $Z$ of $G$ that is $y^{(i)}[u,X]$-accessible from $C[X,u]$ such that $cap(Z)\geq y^{(i)}[X,u]$.
We define $Q_i=(Q_{i-1}\setminus Q_{i-1}[X,u])\cup Z$ (notice that $Z$, $Q_{i-1}[X,u]$, and $C[X,u]$ are $(X,u)$-independent sets that are all compatible with each other as the $Z$ and $Q_{i-1}[X,u]$ are both $y^{(i)}[u,X]$-accessible from $C[X,u]$; this also implies that $Q_i$ is compatible with $C[u,B]$). It follows from our earlier observations that for each $B'\in\beta(u,B)$, $cap(Q_i[B',u])\geq y^{(i)}[B',u]$ and that $Q_i[B',u]$ is $y^{(i)}[u,B']$-accessible from $C[B',u]$. It remains to be shown that $Q_i$ is $y^{(i)}[B,u]$-accessible from $C[u,B]$.

Consider any $(B,u)$-independent set $R$ of $G$ that is compatible with $Q_{i-1}$, and hence is also compatible with $C[u,B]$, such that $cap(R)\geq y^{(i)}[B,u]$.
We will be done if we show that $Q_i\cup R$ is reachable from $C[u,B]\cup R$.
Since $Q_{i-1}$ is $y^{(i)}[B,u]$-accessible from $C[u,B]$ (since $y^{(i)}[B,u]\geq y^{(i-1)}[B,u]$), we know that the independent set $D=Q_{i-1}\cup R$ of $G$ is reachable from $C[u,B]\cup R$.
Notice that $cap(D[B,u])=cap(R)\geq y^{(i)}[B,u]\geq y^{(i-1)}[B,u]$. Also, for any $B'\in\beta(u,B)$, we have $cap(D[B',u])=cap(Q_{i-1}[B',u])\geq y^{(i-1)}[B',u]$. Moreover, for all $B'\in\beta(u,B)\setminus\{X\}$, we have $cap(D[B',u])=cap(Q_{i-1}[B',u])\geq y^{(i-1)}[B',u]=y^{(i)}[B',u]$. We now define a $(u,X)$-independent set $Q'$ of $G$ that is compatible with $D[X,u]$ in such a way that $cap(Q')\geq y^{(i)}[u,X]$ and $D[X,u]\cup Q'=Q_{i-1}[X,u]\cup Q'$ is reachable from $D$ and hence also from $C[u,B]\cup R$.

If $y^{(i)}[u,X]=0$, then we simply define $Q'=D[u,X]$. Clearly, we have $cap(Q')\geq y^{(i)}[u,X]=0$ and that $Q_{i-1}[X,u]\cup Q'=D[X,u]\cup Q'=D[X,u]\cup D[u,X]=D$ is reachable from $D$ in this case.
So we assume that $y^{(i)}[u,X]\neq 0$, which implies that $y^{(i)}[u,X]>0$. We now consider two cases.

First, if $u\in D$ (this means that $u\in C$ as well), we again define $Q'=D[u,X]$. Since $u$ is not under attack in $D[B',u]$ for any $B'\in\beta(u,X)$, we can conclude using Lemma~\ref{lem:ua}\ref{it:uatrue} that there does not exist $B'\in\beta(u,X)$ such that $cap(D[B',u])=0$ and $ua(B',u)=True$. Then we have by Definition~\ref{def:cap} that $cap(Q')=cap(D[u,X])=\displaystyle\sum_{B'\in\beta(u,X)}\Big(cap(D[B',u])-ua(B',u)\Big)+ua(u,X)$. Notice that we again have that $Q_{i-1}[X,u]\cup Q'=D$ is reachable from $D$.

Next, we consider the case when $u\notin D$.
Since $y^{(i)}[u,X]=y^{(i-1)}[u,X]$ (from Observation~\ref{obs:uniquep}), we have that $y^{(i-1)}[u,X]>0$ (from our assumption above that $y^{(i)}[u,X]>0$).
This means that for some $j\in\{1,2,\ldots,i-2\}$, line~\ref{line:assign} of Algorithm~\ref{alg:cap} is executed for $(u,X)$ during the $j$-th iteration of the \textbf{while} loop (recall that $y^{(i-1)}[u,X]=y^{(i)}[u,X]$). It is clear that line~\ref{line:y'uB} must have been executed just before this, which implies that there does not exist distinct $B',B''\in\mathcal{B}_u(G)$ such that $y^{(j)}[B',u]=y^{(j)}[B'',u]=0$ and $ua(B',u)=ua(B'',u)=True$. It follows from Observation~\ref{obs:ymonotone} that there does not exist distinct $B',B''\in\mathcal{B}_u(G)$ such that $y^{(i-1)}[B',u]=y^{(i-1)}[B'',u]=0$ and $ua(B',u)=ua(B'',u)=True$.
This further implies that
there does not exist distinct $B',B''\in\mathcal{B}_u(G)$ such that $cap(D[B',u])=cap(D[B'',u])=0$ and $ua(B',u)=ua(B'',u)=True$. We define $Q'$ to be the $(u,X)$-independent set of $G$ given by Lemma~\ref{lemma:blockshift} that is $\gamma$-accessible from $D[u,X]$ such that
$cap(Q')=\sum\limits_{\forall B'\in\mathcal{B}_u[u,X]}\Big(cap(D[B', u])-ua(B', u)\Big)+ua(u,X)$, where $\gamma=\min\{cap(D[X,u]),1\}$. 
This means that $Q'$ is $cap(D[X,u])$-accessible from $D[u,X]$, which implies that the independent set $D[X,u]\cup Q'=Q_{i-1}[X,u]\cup Q'$ is reachable from $D=D[X,u]\cup D[u,X]$.

We now have that in both cases, $Q'$ is a $(u,X)$-independent set of $G$ that is compatible with $D[X,u]$ such that $D[X,u]\cup Q'=Q_{i-1}[X,u]\cup Q'$ is reachable from $D$, and hence also from $C[u,B]\cup R$. Further we also have that in both cases $cap(Q')=\sum\limits_{\forall B'\in\mathcal{B}_u[u,X]}\Big(cap(D[B', u])-ua(B', u)\Big)+ua(u,X)$.
By our assumption that $y^{(i)}[u,X]>0$, there exists a maximum integer $j$ in $\{1,2,\ldots,i-2\}$ such that line~\ref{line:assign} gets executed for $(u,X)$ in the $j$-th iteration of the \textbf{while} loop in Algorithm~\ref{alg:cap}. Clearly, the value of $y'$ at that point is $\sum\limits_{\forall B'\in \mathcal{B}_u[u,X]}\Big(y^{(j)}[B',u]-ua(B',u)\Big)+ua(u,X)$, and this is also equal to $y^{(i)}[u,X]$ by our choice of $j$. Then we have by Observation~\ref{obs:ymonotone} that $0<y^{(i)}[u,X]\leq\sum\limits_{\forall B'\in \mathcal{B}_u[u,X]}\Big(y^{(i)}[B',u]-ua(B',u)\Big)+ua(u,X)$. Then we have that $cap(Q')\geq\sum\limits_{\forall B'\in\mathcal{B}_u[u,X]}\Big(y^{(i)}[B', u]-ua(B', u)\Big)+ua(u,X)\geq y^{(i)}[u,X]=y^{(i-1)}[u,X]$, as required.

Now since $Q_{i-1}[X,u]$ is $y^{(i-1)}[u,X]$-accessible from $C[X,u]$, it follows that the independent set $C[X,u]\cup Q'$ of $G$ is reachable from $Q_{i-1}[X,u]\cup Q'$ and hence also from $C[u,B]\cup R$. Recall that $Z$ is an $(X,u)$-independent set of $G$ that is $y^{(i)}[u,X]$-accessible from $C[X,u]$ such that $cap(Z)\geq y^{(i)}[X,u]$. Thus $Z\cup Q'$ is an independent set of $G$ that is reachable from $C[X,u]\cup Q'$, and hence also reachable from $C[u,B]\cup R$. Notice that the fact that $y^{(i)}[X,u]\neq y^{(i-1)}[X,u]$ implies that $y^{(i)}[X,u]\geq 1$, which means that $cap(Z)\geq 1$. Since $D[u,X]$ is $1$-accessible from $Q'$ (notice that by the way we defined $Q'$, we have either $Q'=D[u,X]$ or that $Q'$ is $\gamma$-accessible from $D[u,X]$ for some $\gamma\leq 1$), it follows that $D'=Z\cup D[u,X]$ is reachable from $Z\cup Q'$, and therefore also from $C[u,B]\cup R$. Since $D'[B,u]=D[B,u]=R$ and $D'[u,B]=(D[u,B]\setminus D[X,u])\cup Z=(Q_{i-1}\setminus Q_{i-1}[X,u])\cup Z=Q_i$, we have that $Q_i\cup R$ is reachable from $C[u,B]\cup R$, and hence we are done. This completes the proof of the claim.

\begin{claim}\label{clm:Buacc}
Suppose that $p=(B,u)$ for some $u\in V_{cut}(G)$ and $B\in\mathcal{B}_u(G)$.
For all $i\in\{1,2,\ldots,t\}$, there exists a $(B,u)$-independent set $Q_i$ of $G$ that is $y^{(i)}[u,B]$-accessible from $C[B,u]$ such that $B\cap Q_i=B\cap C[B,u]$, and for each $v\in\kappa(B,u)$, $Q_i[v,B]$ is $y^{(i)}[B,v]$-accessible from $C[v,B]$ and $cap(Q_i[v,B])\geq y^{(i)}[v,B]$.
\end{claim}
As in the previous claim, we use induction on $i$. For $i=1$, we let $Q_1=C[B,u]$, and we are done as before. For the inductive step, we assume that $i>1$ and that there exists a $(B,u)$-independent set $Q_{i-1}$ of $G$ that is $y^{(i-1)}[u,B]$-accessible from $C[B,u]$ such that $B\cap Q_{i-1}=B\cap C[B,u]$, and for each $v\in\kappa(B,u)$, $cap(Q_{i-1}[v,B])\geq y^{(i-1)}[v,B]$ and $Q_{i-1}[v,B]$ is $y^{(i-1)}[B,v]$-accessible from $C[v,B]$. If $y^{(i)}[v,B]=y^{(i-1)}[v,B]$ for all $v\in\kappa(B,u)$, then we simply let $Q_i=Q_{i-1}$, and we are done by the induction hypothesis. So let us assume that there exists $x\in\kappa(B,u)$ such that $y^{(i)}[x,B]\neq y^{(i-1)}[x,B]$.

Since $d(x,B)<d(B,u)$, we have by ($*$) that there exists an $(x,B)$-independent set $Z$ of $G$ that is $y^{(i)}[B,x]$-accessible from $C[x,B]$ such that $cap(Z)\geq y^{(i)}[x,B]$. We define $Q_i=(Q_{i-1}\setminus Q_{i-1}[x,B])\cup Z$. Clearly, we have that for each $v\in\kappa(B,u)$, $cap(Q_i[v,B])\geq y^{(i)}[v,B]$ and that $Q_i[v,B]$ is $y^{(i)}[B,v]$-accessible from $C[v,B]$. Also, it can be seen that since $B\cap Q_{i-1}=B\cap C[B,u]$ and $Z$ is compatible with $C[x,B]$, we have $B\cap Q_i=B\cap Q_{i-1}=B\cap C[B,u]$. It remains to be shown that $Q_i$ is $y^{(i)}[u,B]$-accessible from $C[B,u]$.

Consider any $(u,B)$-independent set $R$ of $G$ that is compatible with $Q_{i-1}$, and hence also with $C[u,B]$, such that $cap(R)\geq y^{(i)}[u,B]$.
We will be done if we show that $Q_i\cup R$ is reachable from $C[B,u]\cup R$.
Since $Q_{i-1}$ is $y^{(i)}[u,B]$-accessible from $C[B,u]$ (as $y^{(i)}[u,B]\geq y^{(i-1)}[u,B]$), we know that the independent set $D=Q_{i-1}\cup R$ of $G$ is reachable from $C[B,u]\cup R$. Notice that $cap(D[u,B])=cap(R)\geq y^{(i)}[u,B]\geq y^{(i-1)}[u,B]$.
Clearly, we have $B\cap D[B,u]=B\cap Q_{i-1}=B\cap C[B,u]$, and for each $v\in\kappa(B,u)$, we have $cap(D[v,B])=cap(Q_{i-1}[v,B])\geq y^{(i-1)}[v,B]$. Moreover, for each $v\in\kappa(B,u)\setminus\{x\}$, we have $cap(D[v,B])=cap(Q_{i-1}[v,B])\geq y^{(i-1)}[v,B]=y^{(i)}[v,B]$.
This means that $cap(D[B,x])=\sum\limits_{\forall v\in\kappa(B,x)} cap(D[v,B]) - |B\cap\interior{D[B,x]}| + ua(B,x)\geq \sum\limits_{\forall v\in\kappa(B,x)} y^{(i)}[v,B] - |B\cap\interior{D[B,x]}| + ua(B,x)$.
Since $B\cap D[B,u]=B\cap C[B,u]$, we have $B\cap\interior{D[B,x]}=B\cap\interior{C[B,x]}$, which implies that
$cap(D[B,x])\geq\sum\limits_{\forall v\in\kappa(B,x)} y^{(i)}[v,B] - |B\cap\interior{C[B,x]}| + ua(B,x)$.
If $y^{(i)}[B,x]>0$, there exists a maximum integer $j$ in $\{1,2,\ldots,i-2\}$ such that line~\ref{line:assign} gets executed for $(B,x)$ during the $j$-th iteration of the \textbf{while} loop of Algorithm~\ref{alg:cap}. Clearly, the value of $y'$ at that point is equal to $\sum\limits_{\forall v\in\kappa(B,x)} y^{(j)}[v,B] - |B\cap\interior{C[B,x]}| + ua(B,x)\leq\sum\limits_{\forall v\in\kappa(B,x)} y^{(i)}[v,B] - |B\cap\interior{C[B,x]}| + ua(B,x)$, and by our choice of $j$, it is also equal to $y^{(i)}[B,x]$. Thus we have that either $y^{(i)}[B,x])=0$ or $0<y^{(i)}[B,x]\leq\sum\limits_{\forall v\in\kappa(B,x)} y^{(i)}[v,B] - |B\cap\interior{C[B,x]}| + ua(B,x)$. We now have using Lemma~\ref{lem:ua}\ref{it:capnonnegative} that in any case,
$cap(D[B,x])\geq y^{(i)}[B,x]\geq y^{(i-1)}[B,x]$.
Since $Q_{i-1}[x,B]$ is $y^{(i-1)}[B,x]$-accessible from $C[x,B]$, it follows that the independent set $C[x,B]\cup D[B,x]$ of $G$ is reachable from $Q_{i-1}[x,B]\cup D[B,x]=D$ and hence also from $C[B,u]\cup R$. Recall that $Z$ is an $(x,B)$-independent set of $G$ that is $y^{(i)}[B,x]$-accessible from $C[x,B]$ such that $cap(Z)\geq y^{(i)}[x,B]$. Thus $D'=Z\cup D[B,x]$ is an independent set of $G$ that is reachable from $C[x,B]\cup D[B,x]$, and hence also reachable from $C[B,u]\cup R$. Since $D'[u,B]=D[u,B]=R$ and $D'[B,u]=(D[B,u]\setminus D[x,B])\cup Z=(Q_{i-1}\setminus Q_{i-1}[x,B])\cup Z=Q_i$, we have that $Q_i\cup R=D'$ is reachable from $C[B,u]\cup R$, and hence we are done. This completes the proof of the claim.
\bigskip

We are now ready to complete the proof of the lemma. If $y^{(i)}[p]=0$, then we can simply let $Q=C[p]$, and we would be done. So we assume that $y^{(i)}[p]\neq 0$. Suppose that $p=(u,B)$. By Claim~\ref{clm:uBacc}, we have that there exists a $(u,B)$-independent set $Q_i$ of $G$ that is $y^{(i)}[B,u]$-accessible from $C[u,B]$ such that for all $B'\in\beta(u,B)$, $cap(Q_i[B',u])\geq y^{(i)}[B',u]$. Consider any $(B,u)$-independent set $R$ of $G$ that is compatible with $Q_i$, and hence also compatible with $C[u,B]$, such that $cap(R)\geq y^{(i)}[B,u]$. We know that $D=Q_i\cup R$ is reachable from $C[u,B]\cup R$. We now consider two cases depending upon whether $u\in C$ or $u\notin C$, and in each case we define a $(u,B)$-independent set $Q$ of $G$ which has the property that $cap(Q)=\sum\limits_{\forall B'\in\beta(u,B)}\Big(cap(D[B', u])-ua(B', u)\Big)+ua(u,B)$ and $Q\cup R$ is reachable from $D$.

If $u\in C$, then we define $Q=Q_i$. Since $Q=Q_i$ is compatible with $C[u,B]$, we know that $u\in Q$ and therefore $u$ is not under attack in $Q$. By Lemma~\ref{lem:ua}\ref{it:uatrue}, we now have that there does not exist $B'\in\beta(u,B)$ such that $cap(B',u)=0$ and $ua(B',u)=True$. Then we have by Definition~\ref{def:cap} that $cap(Q)=\sum\limits_{\forall B'\in\beta(u,B)}\Big(cap(D[B', u])-ua(B', u)\Big)+ua(u,B)$. Clearly, we have that $Q\cup R=Q_i\cup R=D$ is reachable from $D$.

Suppose that $u\notin C$. If there exists $B',B''\in\mathcal{B}_u(G)$ such that $cap(D[B',u])=cap(D[B'',u])=0$ and $ua(B',u)=ua(B'',u)=True$, then we have that $y^{(i)}[B',u]=y^{(i)}[B'',u]=0$, which implies (by Algorithm~\ref{alg:cap}) that $y^{(i)}[u,B]=0$, which contradicts our assumption that $y^{(i)}[p]\neq 0$. So we can use Lemma~\ref{lemma:blockshift} to conclude that there exists a $(u,B)$-independent set $Q$ of $G$ that is $\gamma$-accessible from $D[u,B]$ such that $cap(Q)=\sum\limits_{\forall B'\in\beta(u,B)}\Big(cap(D[B', u])-ua(B', u)\Big)+ua(u,B)$, where $\gamma=\min\{cap(D[B,u]),1\}$.
Since $D[B,u]=R$, we have that $Q\cup R$ is reachable from $D[u,B]\cup R=D$.

Thus we have that in both cases, $cap(Q)=\sum\limits_{\forall B'\in\beta(u,B)}\Big(cap(D[B', u])-ua(B', u)\Big)+ua(u,B)$ and $Q\cup R$ is reachable from $D$, and hence also reachable from $C[u,B]\cup R$. This implies that $Q$ is $y^{(i)}[B,u]$-accessible from $C[u,B]$. Finally, we have $cap(Q)=\sum\limits_{\forall B'\in\beta(u,B)}\Big(cap(D[B', u])-ua(B', u)\Big)+ua(u,B)=\sum\limits_{\forall B'\in\beta(u,B)}\Big(cap(Q_i[B',u])-ua(B', u)\Big)+ua(u,B)\geq\sum\limits_{\forall B'\in\beta(u,B)}\Big(y^{(i)}[B', u]-ua(B', u)\Big)+ua(u,B)$. From Algorithm~\ref{alg:cap} and our assumption that $y^{(i)}[u,B]\neq 0$, we have that $0<y^{(i)}[u,B]\leq\sum_{B'\in\beta(u,B)} (y^{(i)}[B',u]-ua(B',u))+ua(u,B)$. It follows that $cap(Q)\geq y^{(i)}[u,B]$, and we are done.

Finally, suppose that $p=(B,u)$. We define $Q$ to be the $(B,u)$-independent set $Q_i$ of $G$ that can be obtained by Claim~\ref{clm:Buacc}. Thus we have that $Q$ is $y^{(i)}[u,B]$-accessible from $C[B,u]$, $B\cap Q=B\cap C[B,u]$, and that for each $v\in\kappa(B,u)$, $cap(Q[v,B])\geq y^{(i)}[v,B]$. Notice that $B\cap\interior{Q}=B\cap\interior{C[B,u]}$.
Then we have from Definition~\ref{def:cap} that $cap(Q)=\sum\limits_{v\in\kappa(B,u)} cap(Q[v,B])-|B\cap\interior{Q}|+ua(B,u)\geq\sum\limits_{v\in\kappa(B,u)} y^{(i)}[v,B]-|B\cap\interior{C[B,u]}|+ua(B,u)$. As before, from our assumption that $y^{(i)}[B,u]\neq 0$ and Algorithm~\ref{alg:cap}, it follows that $0<y^{(i)}[B,u]\leq\sum\limits_{v\in\kappa(B,u)} y^{(i)}[v,B]-|B\cap\interior{C[B,u]}|+ua(B,u)$. We thus have that $cap(Q)\geq y^{(i)}[B,u]$ and we are done. This completes the proof.
\end{proof}

\begin{lemma}\label{lem:thereexists}
Let $p\in\mathcal{P}$.
There exists an independent set $C'$ of $G$ that is reachable from $C$ such that $cap(C'[p])\geq\mathbf{x}[p]$ and $|\interior{C'[p]}|=|\interior{C[p]}|$.
\end{lemma}
\begin{proof}
Note that as $\mathbf{x}[p]=y^{(t)}[p]$, we will be done if we prove the following stronger statement:
\medskip

\textit{Let $p\in\mathcal{P}$. For each $i\in\{1,2,\ldots,t\}$, there exists an independent set $C^{(i)}$ that is reachable from $C$ such that $C^{(i)}[p]$ is $y^{(i)}(\overline{p})$-accessible from $C[p]$, $cap(C^{(i)}[p])\geq y^{(i)}[p]$, $C^{(i)}[\overline{p}]$ is $y^{(i)}[p]$-accessible from $C[\overline{p}]$, and $cap(C^{(i)}[\overline{p}])\geq y^{(i)}[\overline{p}]$.}
\medskip

We prove this statement by induction on $i$. For $i=1$, we have $y^{(i)}[p]=y^{(i)}[\overline{p}]=0$, so we can simply let $C^{(1)}=C$. Suppose that $i>1$ and that there is an independent set $C^{(i-1)}$ that is reachable from $C$ such that $C^{(i-1)}[p]$ is $y^{(i-1)}[\overline{p}]$-accessible from $C[p]$, $cap(C^{(i-1)}[p])\geq y^{(i-1)}[p]$, $C^{(i-1)}[\overline{p}]$ is $y^{(i-1)}[p]$-accessible from $C[\overline{p}]$, and $cap(C^{(i-1)}[\overline{p}])\geq y^{(i-1)}[\overline{p}]$. If $y^{(i)}[p]=y^{(i-1)}[p]$ and $y^{(i)}[\overline{p}]=y^{(i-1)}[\overline{p}]$, then we simply let $C^{(i)}=C^{(i-1)}$ and we are done. So we can assume by symmetry that $y^{(i)}[p]>y^{(i-1)}[p]$ and $y^{(i)}[\overline{p}]=y^{(i-1)}[\overline{p}]$. By Lemma~\ref{lem:evolution}, we have that there exists a $p$-independent set $Q$ that is $y^{(i)}[\overline{p}]$-accessible from $C[p]$ such that $cap(Q)\geq y^{(i)}[p]$. Since $C^{(i-1)}[p]$ is also $y^{(i)}[\overline{p}]$-accessible from $C[p]$, we have that $Q$ is $y^{(i)}[\overline{p}]$-accessible from $C^{(i-1)}[p]$. We define $C^{(i)}=Q\cup C^{(i-1)}[\overline{p}]$. Since $cap(C^{(i-1)}[\overline{p}])\geq y^{(i-1)}[\overline{p}]=y^{(i)}[\overline{p}]$, we can conclude that $C^{(i)}$ is reachable from $C^{(i-1)}$. Clearly, $C^{(i)}[p]=Q$ is $y^{(i)}[\overline{p}]$-accessible from $C[p]$, $cap(C^{(i)}[p])=cap(Q)\geq y^{(i)}[p]$, $C^{(i)}[\overline{p}]=C^{(i-1)}[\overline{p}]$ is $y^{(i)}[p]$-accessible from $C[\overline{p}]$ (since $y^{(i)}[p]>y^{(i-1)}[p]$) and $cap(C^{(i)}[\overline{p}])=cap(C^{(i-1)}[\overline{p}])\geq y^{(i-1)}[\overline{p}]=y^{(i)}[\overline{p}]$. This completes the proof.
\end{proof}

\begin{lemma}\label{lem:forevery}
For every $p\in\mathcal{P}$ and every independent set $C'$ of $G$ that is reachable from $C$, we have $\mathbf{x}[p]\geq cap(C'[p])+|\interior{C'[p]}|-|\interior{C[p]}|$.
\end{lemma}
\begin{proof}
For any independent set $C'$ of $G$ that is reachable from $C$, define $dist(C,C')$ to be the minimum integer $k$ for which there exist independent sets $C_0,C_1,\ldots,C_k$ such that $C=C_0\rightarrow C_1\rightarrow\cdots\rightarrow C_k=C'$. Suppose for the sake of contradiction that there exists some independent set $C'$ of $G$ that is reachable from $C$ such that for some $p\in\mathcal{P}$, $\mathbf{x}[p]<cap(C'[p])+|\interior{C'[p]}|-|\interior{C[p]}|$. We assume that $C'$ is such an independent set for which $dist(C,C')$ is as small as possible; i.e. for all independent sets $C''$ of $G$ such that $dist(C,C'')<dist(C,C')$, we have that for all $p\in\mathcal{P}$, $\mathbf{x}[p]\geq cap(C''[p])+|\interior{C''[p]}|-|\interior{C[p]}|$. Note that it is possible that $C'=C$.

Let $p\in\mathcal{P}$ such that $\mathbf{x}[p]<cap(C'[p])+|\interior{C'[p]}|-|\interior{C[p]}|$ having smallest depth; i.e. for all $p'\in\mathcal{P}$ such that $d(p')<d(p)$, we have $\mathbf{x}[p']\geq cap(C'[p'])+|\interior{C'[p']}|-|\interior{C[p']}|$.
Suppose that $d(p)=0$.
Then it follows from Definition~\ref{def:depth} that $p=(B,u)$ for some $u\in V_{cut}(G)$ and $B\in\mathcal{B}_u(G)$, and that $G[p]$ is a complete graph. Then by Lemma~\ref{lem:xBu}, we know that $\mathbf{x}[p]=\mathbf{x}[B,u]=ua(B,u)-|B\cap\interior{C[B,u]}|=1-|\interior{C[B,u]}|$. By Observation~\ref{obs:capcomplete}, we have $cap(C'[B,u])=1-|\interior{C'[B,u]}|$, which can be rewritten as $cap(C'[B,u])+|\interior{C'[B,u]}|=1$. We thus have that $\mathbf{x}[B,u]=cap(C'[B,u])+|\interior{C'[B,u]}|-|\interior{C[B,u]}|$, which contradicts the assumption that $\mathbf{x}[p]<cap(C'[p])+|\interior{C'[p]}|-|\interior{C[p]}|$.

So we can assume that $d(p)>0$.

Suppose that $p=(B,u)$ for some $u\in V_{cut}(G)$ and $B\in\mathcal{B}_u(G)$. By Lemma~\ref{lem:xBu}, we have that $\mathbf{x}[B,u]=\displaystyle\sum_{\forall u'\in \kappa(B,u)} \mathbf{x}[u',B]+ua(B,u)-|B \cap\interior{C[B,u]}|$. Since for all $u'\in\kappa(B,u)$, we have $d(u',B)<d(B,u)$, we have by our choice of $p$ that $\mathbf{x}[u',B]\geq cap(C'[u',B])+|\interior{C'[u',B]}|-|\interior{C[u',B]}|$.
Thus we have:
\begin{eqnarray*}
\mathbf{x}[B,u]&\geq&\sum_{\forall u'\in \kappa(B,u)} \Big(cap(C'[u',B])+|\interior{C'[u',B]}|-|\interior{C[u',B]}|\Big)+ua(B,u)-|B \cap\interior{C[B,u]}|\\
&=&|\interior{C'[B,u]}|-|B\cap\interior{C'[B,u]}|-\big(|\interior{C[B,u]}|-|B\cap\interior{C[B,u]}|\big)\\&&\quad\quad\quad+ua(B,u)-|B \cap\interior{C[B,u]}|+\sum_{\forall u'\in \kappa(B,u)} cap(C'[u',B])\mbox{\hspace{0.2in} (By Observation~\ref{obs:sum})}\\
&=&cap(C'[B,u])+|\interior{C'[B,u]}|-|\interior{C[B,u]}| \mbox{\hspace{1.6in} (By Definition~\ref{def:cap})}
\end{eqnarray*}
which contradicts the assumption that $\mathbf{x}[p]<cap(C'[p])+|\interior{C'[p]}|-|\interior{C[p]}|$.

Next, suppose that $p=(u,B)$ for some $u\in V_{cut}(G)$ and $B\in\mathcal{B}_u(G)$.

Suppose first that there exist $B',B''\in\mathcal{B}_u(G)$ such that $\mathbf{x}[B',u]=\mathbf{x}[B'',u]=0$ and $ua(B',u)=ua(B'',u)=True$. Then we have by Lemma~\ref{lem:xuB} that $\mathbf{x}[u,B]=0$.
Note that at least one of $B',B''$ is in $\beta(u,B)$. We assume without loss of generality that $B'\in\beta(u,B)$.
If $u\in C$, then $u$ is not under attack in $C[B',u]$. As $ua(B',u)=True$, we then have by Lemma~\ref{lem:ua}\ref{it:uatrue} that $cap(C[B',u])>0$. Since $d(B',u)<d(p)$, we have by our choice of $p$ that $\mathbf{x}[B',u]\geq cap(C[B',u])+|\interior{C[B',u]}|-|\interior{C[B',u]}|=cap(C[B',u])>0$, which contradicts the fact that $\mathbf{x}[B',u]=0$. So we can conclude that $u\notin C$.


Since $C'$ is reachable from $C$, there exist independent sets $C_0,C_1,\ldots,C_k$ of $G$ such that $k=dist(C,C')$ and $C=C_0\rightarrow C_1\rightarrow\cdots\rightarrow C_k=C'$. If $u\notin C_0\cup C_1\cup\cdots\cup C_k$, then by Corollary~\ref{cor:ureachable}\ref{it:cardinality}, we have that $|\interior{C'[u,B]}|=|\interior{C[u,B]}|$. Then $cap(C'[u,B])+|\interior{C'[u,B]}|-|\interior{C[u,B]}|=cap(C'[u,B])$, which means that we will be done if we show that $cap(C'[u,B])=0$. Since $u\notin C_0\cup C_1\cup\cdots\cup C_k$, we have by Corollary~\ref{cor:ureachable}\ref{it:cardinality} that $|\interior{C'[B',u]}|=|\interior{C[B',u]}|$ and $|\interior{C'[B'',u]}|=|\interior{C[B'',u]}|$.
Recall that since $d(B',u)<d(p)$, we have by our choice of $p$ that $0=\mathbf{x}[B',u]\geq cap(C'[B',u])+|\interior{C'[B',u]}|-|\interior{C[B',u]}|$. We can then conclude using Lemma~\ref{lem:ua}\ref{it:capnonnegative} that $cap(C'[B',u])=0$. We now have by Definition~\ref{def:cap} that $cap(C'[u,B])=0$, and so we are done. We can thus assume that $u\in C_0\cup C_1\cup\cdots\cup C_k$. Then there exists a smallest integer $i$ in $\{0,1,\ldots,k\}$ such that $u\in C_i$. By our earlier observation that $u\notin C$, we have that $i>0$ (this also means that $C'\neq C$). By our choice of $i$, we have $u\notin C_0\cup C_1\cup\cdots\cup C_{i-1}$ and since $C_{i-1}\rightarrow C_i$, we also have that $u$ is not under attack in at least one of $C_{i-1}[B',u]$ or $C_{i-1}[B'',u]$. Let $B^\star\in\{B',B''\}$ be such that that $u$ is not under attack in $C_{i-1}[B^\star,u]$. Then by Lemma~\ref{lem:ua}\ref{it:uatrue}, we have that $cap(C_{i-1}[B^\star,u])>0$. Clearly, $C_{i-1}$ is reachable from $C$, and $dist(C,C_{i-1})=i-1<k=dist(C,C')$. Then we have by our choice of $C'$ that $0=\mathbf{x}[B^\star,u]\geq cap(C_{i-1}[B^\star,u])+|\interior{C_{i-1}[B^\star,u]}|-|\interior{C[B^\star,u]}|$, which implies that $|\interior{C_{i-1}[B^\star,u]}|<|\interior{C[B^\star,u]}|$. But since $u\notin C_0\cup C_1\cup\cdots\cup C_{i-1}$, we have by Corollary~\ref{cor:ureachable}\ref{it:cardinality} that $|\interior{C_{i-1}[B^\star,u]}|=|\interior{C[B^\star,u]}|$, which is a contradiction.

So we can assume that there does not exist $B',B''\in\mathcal{B}_u(G)$ such that $\mathbf{x}[B',u]=\mathbf{x}[B'',u]=0$ and $ua(B',u)=ua(B'',u)=True$. Then by Lemma~\ref{lem:xuB}, we have that $\mathbf{x}[u,B]=\sum\limits_{\forall B'\in\beta(u,B)}\Big(\mathbf{x}[B', u]-ua(B', u)\Big)+ua(u,B)$. For all $B'\in\beta(u,B)$, since $d(B',u)<d(p)$, we have by our choice of $p$ that $\mathbf{x}[B',u]\geq cap(C'[B',u])+|\interior{C'[B',u]}|-|\interior{C[B',u]}|$.
We thus get
\begin{eqnarray*}
\mathbf{x}[u,B]&\geq&\sum\limits_{\forall B'\in\beta(u,B)}\Big(cap(C'[B',u])+|\interior{C'[B', u]}|-|\interior{C[B',u]}|-ua(B', u)\Big)+ua(u,B)\\
&=&|\interior{C'[u,B]}|-|\interior{C[u,B]}|+\sum\limits_{\forall B'\in\beta(u,B)}\Big(cap(C'[B',u])-ua(B', u)\Big)+ua(u,B)\\&&\mbox{\hspace{3in} (By Observation~\ref{obs:sum})}
\end{eqnarray*}
Recall that $\mathbf{x}[p]<cap(C'[p])+|\interior{C'[p]}|-|\interior{C[p]}|$; i.e. $\mathbf{x}[u,B]<cap(C'[u,B])+|\interior{C'[u,B]}|-|\interior{C[u,B]}|$. This means that $cap(C'[u,B])\neq \sum\limits_{\forall B'\in\beta(u,B)}\Big(cap(C'[B',u])-ua(B', u)\Big)+ua(u,B)$, which implies by Definition~\ref{def:cap} that $cap(C'[u,B])=0$. This implies that $\mathbf{x}[u,B]<|\interior{C'[u,B]}|-|\interior{C[u,B]}|$. Since $\mathbf{x}[u,B]\geq 0$ (by Observation~\ref{obs:ymonotone}), this means that $|\interior{C'[u,B]}|-|\interior{C[u,B]}|>0$, or in other words, $|\interior{C'[u,B]}|>|\interior{C[u,B]}|$.

As $\mathbf{x}[u,B]=\sum\limits_{\forall B'\in\beta(u,B)}\Big(\mathbf{x}[B', u]-ua(B', u)\Big)+ua(u,B)$, we now have by Observation~\ref{obs:sum} that $\sum\limits_{\forall B'\in\beta(u,B)}\Big(\mathbf{x}[B', u]-ua(B', u)\Big)+ua(u,B)<|\interior{C'[u,B]}|-|\interior{C[u,B]}|=\sum\limits_{\forall B'\in\beta(u,B)}\Big(|\interior{C'[B', u]}|-|\interior{C[B',u]}|\Big)$. Recall that for each $B'\in\beta(u,B)$, $\mathbf{x}[B',u]\geq cap(C'[B',u])+|\interior{C'[B',u]}|-|\interior{C[B',u]}|$, and since $cap(C'[B',u])\geq 0$ by Lemma~\ref{lem:ua}\ref{it:capnonnegative}, we get $\mathbf{x}[B',u]\geq |\interior{C'[B',u]}|-|\interior{C[B',u]}|$.
Combining, we get $$\sum\limits_{\forall B'\in\beta(u,B)}\Big(|\interior{C'[B',u]}|-|\interior{C[B',u]}|-ua(B', u)\Big)+ua(u,B)<\sum\limits_{\forall B'\in\beta(u,B)}\Big(|\interior{C'[B', u]}|-|\interior{C[B',u]}|\Big)$$
If $ua(u,B)=False$, then we have from Definition~\ref{def:ua} that $ua(B',u)=False$ for each $B'\in\beta(u,B)$, and we will get a contradiction to the above inequality. So we can assume that $ua(u,B)=True$.
It now follows from the inequality $\sum\limits_{\forall B'\in\beta(u,B)}\Big(\mathbf{x}[B', u]-ua(B', u)\Big)+ua(u,B)<\sum\limits_{\forall B'\in\beta(u,B)}\Big(|\interior{C'[B', u]}|-|\interior{C[B',u]}|\Big)$ that $\sum\limits_{\forall B'\in\beta(u,B)}\Big(\mathbf{x}[B', u]-\big(|\interior{C'[B', u]}|-|\interior{C[B',u]}|\big)-ua(B', u)\Big)<-1$.
Recalling that $\mathbf{x}[B',u]\geq |\interior{C'[B',u]}|-|\interior{C[B',u]}|$ for each $B'\in\beta(u,B)$, we can now conclude that there exist distinct $B',B''\in\beta(u,B)$ such that
$\mathbf{x}[B',u]=|\interior{C'[B',u]}|-|\interior{C[B',u]}|$, $\mathbf{x}[B'',u]=|\interior{C'[B'',u]}|-|\interior{C[B'',u]}|$, and $ua(B',u)=ua(B'',u)=True$. Note that this implies that $cap(C'[B',u])=cap(C'[B'',u])=0$, since we know that $\mathbf{x}[B',u]\geq cap(C'[B',u])+|\interior{C'[B',u]}|-|\interior{C[B',u]}|$ and $\mathbf{x}[B'',u]\geq cap(C'[B'',u])+|\interior{C'[B'',u]}|-|\interior{C[B'',u]}|$, and also that $cap(C'[B',u])\geq 0$ and $cap(C'[B'',u])\geq 0$ by Lemma~\ref{lem:ua}\ref{it:capnonnegative}. Then by Lemma~\ref{lem:ua}\ref{it:uatrue}, we have that $u$ is under attack in both $C'[B',u]$ and $C'[B'',u]$.

Since $C'$ is reachable from $C$, there exist independent sets $C_0,C_1,\ldots,C_k$ of $G$, where $k\geq 0$, such that $C=C_0\rightarrow C_1\rightarrow\cdots\rightarrow C_k=C'$. Since $|\interior{C'[u,B]}|>|\interior{C[u,B]}|$, we have by Corollary~\ref{cor:ureachable}\ref{it:cardinality} that $u\in C_0\cup C_1\cup\cdots\cup C_k$. Let $i$ be the maximum integer in $\{0,1,\ldots,k\}$ such that $u\in C_i$. Clearly, $i<k$ since $u$ is under attack in $C'$. Since $C_{i+1}\rightarrow C_{i+2}\rightarrow\cdots\rightarrow C_k$, and we have by our choice of $i$ that $u\notin C_{i+1}\cup C_{i+2}\cup\cdots\cup C_k$, we can conclude using Corollary~\ref{cor:ureachable}\ref{it:cardinality} that $|\interior{C_{i+1}[B',u]}|=|\interior{C'[B',u]}|$ and $|\interior{C_{i+1}[B'',u]}|=|\interior{C'[B'',u]}|$. Since $C_i\rightarrow C_{i+1}$, we also have that $u$ is not under attack in at least one of $C_{i+1}[B',u]$, $C_{i+1}[B'',u]$. Note that this implies that $i+1<k$, or in other words, $C_{i+1}\neq C'$. We assume without loss of generality that $u$ is not under attack in $C_{i+1}[B',u]$. Then by Lemma~\ref{lem:ua}\ref{it:uatrue}, we have that $cap(C_{i+1}[B',u])>0$. As $C_{i+1}$ is reachable from $C$, and $dist(C,C_{i+1})=i+1<k=dist(C,C')$, we have by our choice of $C'$ that $\mathbf{x}[B',u]\geq cap(C_{i+1}[B',u])+|\interior{C_{i+1}[B',u]}|-|\interior{C[B',u]}|=cap(C_{i+1}[B',u])+|\interior{C'[B',u]}|-|\interior{C[B',u]}|$, which implies that $\mathbf{x}[B',u]>|\interior{C'[B',u]}|-|\interior{C[B',u]}|$. This contradicts the fact that $\mathbf{x}[B',u]=|\interior{C'[B',u]}|-|\interior{C[B',u]}|$.
\end{proof}
\begin{corollary}
For each $p\in\mathcal{P}$, $\mathbf{x}[p]=pot(C,p)$. 
\end{corollary}
\begin{proof}
Clearly, for each independent set $C'$ of $G$ that is reachable from $C$, we have from Lemma~\ref{lem:forevery} that $cap(C'[p])+|\interior{C'[p]}|-|\interior{C[p]}|\leq\mathbf{x}[p]$. It follows that $pot(C,p)\leq\mathbf{x}[p]$. From Lemma~\ref{lem:thereexists}, we have that there exists an independent set $C'$ of $G$ that is reachable from $C$ such that $cap(C'[p])\geq\mathbf{x}[p]$ and $|\interior{C'[p]}|=|\interior{C[p]}|$. Then we have that $cap(C'[p])+|\interior{C'[p]}|-|\interior{C[p]}|=cap(C'[p])\geq\mathbf{x}[p]$. This implies that $pot(C,p)\geq\mathbf{x}[p]$. This completes the proof.
\end{proof}
By the above corollary, we can reword Lemmas~\ref{lem:xBu}, \ref{lem:xuB}, and~\ref{lem:thereexists} as follows.
\begin{corollary}\label{cor:pot}
Let $G$ be any connected block graph and $C$ be an independent set of $G$.
\begin{enumerate}
\renewcommand{\labelenumi}{(\roman{enumi})}
\renewcommand{\theenumi}{(\roman{enumi})}
\item\label{cor:xBu} For any $u\in V_{cut}(G)$ and $B\in\mathcal{B}_u(G)$, $$pot_G(C,(B,u))=\sum\limits_{\forall u'\in\kappa_G(B,u)} pot_G(C,(u',B))+ua_G(B,u)-|B \cap \interior{C[B,u]}|$$
\item\label{cor:xuB} For any $u\in V_{cut}(G)$ and $B\in\mathcal{B}_u(G)$,

if $\exists B',B''\in\mathcal{B}_u(G)$ such that $pot_G(C,(B',u))=pot_G(C,(B'',u))=0$ and $ua_G(B',u)=ua_G(B'',u)=True$, then $$pot_G(C,(u,B))=0$$

and otherwise,
$$pot_G(C,(u,B))=\sum\limits_{\forall B'\in\beta_G(u,B)}\Big(pot_G(C,(B', u))-ua_G(B', u)\Big)+ua_G(u,B)$$
\item\label{cor:thereexists} For each $p\in\mathcal{P}_G$, there exists an independent set $C'$ of $G$ that is reachable from $C$ such that $cap(C'[p])\geq pot_G(C,p)$ and $|\interior{C'[p]}|=|\interior{C[p]}|$.
\end{enumerate}
\end{corollary}
Also, Lemma~\ref{lem:polytime} can now be reworded as:
\begin{theorem}\label{thm:polytime}
Given a connected block graph $G$ and an independent set $C$ of it, the procedure {\sc Compute-potentials}$(G,C)$ runs in time $O(|E(G)|^4)$ and computes the value $pot_G(C,p)$ for each $p\in\mathcal{P}_G$.
\end{theorem}
\section{Rigid vertices}\label{sec:rigid}
Let $G$ be a connected block graph and $C$ be any independent set of $G$.
\begin{definition}\label{def:rigid}
A vertex $u\in V_{cut}(G)$ is said to be \emph{rigid for $C$}, if there exist $B,B'\in\mathcal{B}_u(G)$ such that $pot_G(C,(B,u))=pot_G(C,(B',u))=0$ and $ua_G(B,u)=ua_G(B',u)=True$.
\end{definition}
With the above definition, the following observation follows directly from Corollary~\ref{cor:pot}\ref{cor:xuB} and Definition~\ref{def:ua}.
\begin{observation}\label{obs:rigidua}
If $u\in V_{cut}(G)$ is rigid for $C$, then for all $B\in\mathcal{B}_u(G)$, $ua_G(u,B)=True$ and $pot_G(C,(u,B))=0$.
\end{observation}
\begin{lemma}\label{lem:rigidnotoken}
If $u\in V_{cut}(G)$ is rigid for an independent set $C$ of $G$, then there does not exist any independent set $C'$ of $G$ that is reachable from $C$ for which $u\in C'$.
\end{lemma}
\begin{proof}
Since $u$ is rigid for $C$ we know that there exist $B,B'\in\mathcal{B}_u(G)$ such that $pot_G(C,(B,u))=pot_G(C,(B',u))=0$ and $ua_G(B,u)=ua_G(B',u)=True$.
Suppose that there exists an independent set $C'$ of $G$ that is reachable from $C$ such that $u\in C'$. Let $C_0,C_1,\ldots,C_k$ be independent sets of $G$ such that $C=C_0\rightarrow C_1\rightarrow\cdots\rightarrow C_k=C'$. Let $i$ be the minimum integer in $\{0,1,\ldots,k\}$ such that $u\in C_i$. If $i=0$, then we clearly have $|\interior{C_i[B,u]}|=|\interior{C[B,u]}|$ and $|\interior{C_i[B',u]}|=|\interior{C[B',u]}|$. On the other hand, if $i>0$, then since $u\notin C_0\cup C_1\cup\cdots\cup C_{i-1}$, we know by Corollary~\ref{cor:ureachable}\ref{it:cardinality} that $|\interior{C_{i-1}[B,u]}|=|\interior{C[B,u]}|$ and $|\interior{C_{i-1}[B',u]}|=|\interior{C[B',u]}|$, and since $C_{i-1}\rightarrow C_i$, this implies that either $|\interior{C_i[B,u]}|=|\interior{C[B,u]}|$ or $|\interior{C_i[B',u]}|=|\interior{C[B',u]}|$. Combining the two cases, we assume without loss of generality that $|\interior{C_i[B,u]}|=|\interior{C[B,u]}|$. Since $u\in C_i$, we know that $u$ is not under attack in $C_i[B,u]$. Then since $C_i$ is reachable from $C$, we have by Definition~\ref{def:pot} that $cap(C_i[B,u])+|\interior{C_i[B,u]}|-|\interior{C[B,u]}|\leq pot_G(C,(B,u))=0$, which implies that $cap(C_i[B,u])=0$. As $u$ is not under attack in $C_i[B,u]$, we now have a contradiction to Lemma~\ref{lem:ua}\ref{it:uatrue}. This completes the proof.
\end{proof}
\begin{lemma}\label{lem:nonrigid}
Let $u\in V_{cut}(G)$ such that $u$ is not rigid for $C$ and let $B\in\mathcal{B}_u(G)$. If $u$ is under attack in $C[u,B]$ (i.e. $|N_G(u)\cap C[u,B]|\geq 1$), then there is an independent set $C'$ of $G$ that is $(u,B)$-reachable from $C$ such that $|N_G(u)\cap C'[u,B]|=1$.
\end{lemma}
\begin{proof}
Let $\beta_G[u,B]=\{B_1,B_2,\ldots,B_m\}$.
If there exists $B'\in\beta_G[u,B]$ such that $pot_G(C,(B',u))=0$ and $ua_G(B',u)=True$, then without loss of generality, we let $B_1=B'$. Note that in this case, we have from Definition~\ref{def:pot} that $cap(C[B_1,u])=0$. Since $ua_G(B_1,u)=True$, we now have by Lemma~\ref{lem:ua}\ref{it:uatrue} that $u$ is under attack in $C[B_1,u]$.
On the other hand, if there does not exist $B'\in\beta_G(u,B)$ such that $pot_G(C,(B',u))=0$ and $ua_G(B',u)=True$, then we assume that $B_1$ is a block in $\beta_G[u,B]$ such that $u$ is under attack in $C[B_1,u]$ (such a block exists since $u$ is under attack in $C[u,B]$). Thus we have that in any case, $u$ is under attack in $C[B_1,u]$.

For each $i\in\{2,3,\ldots,m\}$, we say that the block $B_i$ is ``good'' if there exists a $(B_i,u)$-independent set $Q_i$ of $G$ that is strongly accessible from $C[B_i,u]$ such that $u$ is not under attack in $Q_i$.

Suppose that $B_i$ is good for each $i\in\{2,3,\ldots,m\}$.
Then for each $i\in\{2,3,\ldots,m\}$, there exists a $(B_i,u)$-independent set $Q_i$ that is strongly accessible from $C[B_i,u]$ in which $u$ is not under attack. Applying Lemma~\ref{lem:reconfig} (setting $\mathcal{A}=\{B_2,B_3,\ldots,B_m\}$ and $Q_{B_i}=Q_i$ for each $i\in\{2,3,\ldots,m\}$), we get that the independent set $C'=C[B_1,u]\cup Q_2\cup Q_3\ldots Q_m\cup C[B,u]$ is $(\{B_2,B_3,\ldots,B_m\},u)$-reachable from $C$. Then clearly, $|N_G(u)\cap C'[u,B]|=1$, and by Observation~\ref{obs:reachsubset}, we have that $C'$ is $(u,B)$-reachable from $C$.

Next, we consider the case when there exists $i\in\{2,3,\ldots,m\}$ such that $B_i$ is not good; i.e. there does not exist a $(B_i,u)$-independent set $Q_i$ of $G$ that is strongly accessible from $C[B_i,u]$ such that $u$ is not under attack in $Q_i$. Without loss of generality, let $B_2$ be one such block. Since $u$ is not rigid for $C$, there exists at most one $B'\in\beta_G[u,B]$ such that $pot_G(C,(B',u))=0$ and $ua_G(B',u)=True$. Then, from our choice of $B_1$, it is clear that $pot_G(C,(B_2,u))>0$. Now, by Corollary~\ref{cor:pot}\ref{cor:thereexists}, there exists an independent set $S$ of $G$ that is reachable from $C$ such that $cap(S[B_2,u])\geq pot_G(C,(B_2,u))>0$. 
Then there exist independent sets $C_0,C_1,\ldots,C_k$ of $G$ such that $C=C_0\rightarrow C_1\rightarrow \cdots \rightarrow C_k=S$, where $k\geq 0$, such that $cap(S[B_2,u])\geq pot_G(C,(B_2,u))>0$. If $u \notin C_0\cup C_1\cup\cdots\cup C_k$, then from Lemma~\ref{lem:ureachable}, we get a $(B_2,u)$-independent set $S[B_2,u]$ of $G$ that is strongly accessible from $C[B_2,u]$ such that $cap(S[B_2,u])>0$. Then we have by Lemma~\ref{lem:uafalse} that there exists a $(B_2,u)$-independent set $Q_2$ that is strongly accessible from $S[B_2,u]$, and therefore also from $C[B_2,u]$, in which $u$ is not under attack. This contradicts our assumption that $B_2$ is not good. So we can assume that $u\in C_0\cup C_1\cup\cdots\cup C_k$. Let $j=\min\{i\in\{0,1,\ldots,k\}\colon u\in C_i\}$. Since $u\in C_j$, we have that $N_G(u)\cap C_j=\emptyset$. Moreover, since $u\notin C_{j-1}$ and $C_{j-1}\rightarrow C_j$, we get that $C_j\setminus C_{j-1}=\{u\}$ and there exists $v\in N_G(u)$ such that $C_{j-1}\setminus C_j=\{v\}$. This implies that $N_G(u)\cap C_{j-1}=\{v\}$. 

Suppose that $v\notin B_2$. Then $C=C_0\rightarrow C_1\rightarrow \cdots \rightarrow C_{j-1}$, $u \notin C_0\cup C_1\cup\cdots\cup C_{j-1}$ and $u$ is not under attack in $C_{j-1}[B_2,u]$. Then by Lemma~\ref{lem:ureachable}, we get that $Q_2=C_{j-1}[B_2,u]$ is a $(B_2,u)$-independent set of $G$ that is strongly accessible from $C[B_2,u]$ such that $u$ is not under attack in $Q_2$. This again contradicts the fact that $B_2$ is not good. So we can assume that $v\in B_2$. Now, since $C=C_0\rightarrow C_1\rightarrow \cdots \rightarrow C_{j-1}$ and $u \notin C_0\cup C_1\cup\cdots\cup C_{j-1}$, by Corollary~\ref{cor:ureachable}\ref{it:blocks} (setting $\mathcal{A}=\beta_G(u,B)=\{B_1,B_2,\ldots,B_m\}$ and $Q_{B_i}=C_{j-1}[B_i,u]$ for each $i\in\{1,2,\ldots,m\}$), we have that $C'=C[B,u]\cup C_{j-1}[B_1,u]\cup C_{j-1}[B_2,u]\cup\cdots\cup C_{j-1}[B_m,u]=C[B,u]\cup C_{j-1}[u,B]$ is an independent set of $G$ that is $(u,B)$-reachable from $C$. Since $|N_G(u)\cap C'[u,B]|=|N_G(u)\cap C_{j-1}[u,B]|=1$, we are done.
\end{proof}
\begin{lemma}\label{lem:rigidmatch}
Let $C_1$ and $C_2$ be two independent sets of $G$ and let $W_1$ and $W_2$ be the set of cut-vertices of $G$ that are rigid for $C_1$ and $C_2$ respectively. If $C_2$ is reachable from $C_1$, then $W_1=W_2$. 
\end{lemma}
\begin{proof}
Suppose for the sake of contradiction that $C_2$ is reachable from $C_1$ and $W_1\neq W_2$. We can assume without loss of generality that there exists $u\in W_1\setminus W_2$.
Then we know that there exists $B\in\mathcal{B}_u(G)$ such that $ua_G(B,u)=True$, $pot_G(C_1,(B,u))=0$, and $pot_G(C_2,(B,u))\geq 1$. By Corollary~\ref{cor:pot}\ref{cor:thereexists}, we have that there is an independent set $C'_2$ that is reachable from $C_2$ such that $cap(C'_2[B,u])\geq pot_G(C_2,(B,u))\geq 1$. Since $C_2$ is reachable from $C_1$, we have that $C'_2$ is also reachable from $C_1$. Then by Definition~\ref{def:pot}, we get that $cap(C'_2[B,u])-|\interior{C'_2[B,u]}|+|\interior{C_1[B,u]}|\leq pot_G(C_1,(B,u))=0$, which implies that $|\interior{C'_2[B,u]}|-|\interior{C_1[B,u]}|\geq 1$. Let $D_0,D_1,\ldots,D_k$ be independent sets of $G$ such that $C_1=D_0\rightarrow D_1\rightarrow\cdots\rightarrow D_k=C'_2$. As $|\interior{C'_2[B,u]}|\neq |\interior{C_1[B,u]}|$, we have from Corollary~\ref{cor:ureachable}\ref{it:cardinality} that there exists some $i\in\{0,1,\ldots,k\}$ such that $u\in D_i$. Then $D_i$ is an independent set of $G$ that is reachable from $C_1$ and contains the vertex $u$ that is rigid for $C_1$. This contradicts Lemma~\ref{lem:rigidnotoken}.
\end{proof}
\section{Restricting to subgraphs}\label{sec:removerigid}
In this section, we show that given a connected block graph $G$ and an independent set $C$ of it, there are certain kinds of induced subgraphs $H$ of $G$ having the property that each potential of $H$ has the same value as it had in $G$.

For this section, we assume that $G$ is a connected block graph and $C$ is an independent set of $G$.
Let $a\in V_{cut}(G)$ and $A\in\mathcal{B}_a(G)$.
Let $H=G-V(G[a,A])$. It is easy to see that $H$ is a connected block graph.
Observe that $\mathcal{B}(H)=(\mathcal{B}(G[A,a])\setminus\{A\})\cup (A\setminus\{a\})$ if $|A|>2$ and $\mathcal{B}(H)=\mathcal{B}(G[A,a])\setminus\{A\}$ if $|A|=2$.
Observe that $V_{cut}(H)\subseteq V_{cut}(G)$.
We now define a function $f_{H,G}:\mathcal{B}(H)\rightarrow\mathcal{B}(G)$.
For each $B\in\mathcal{B}(H)$, we define $f_{H,G}(B)\in\mathcal{B}(G)$ as follows.
If $|A|=2$, then we simply define $f_{H,G}(B)=B$ for all $B\in\mathcal{B}(H)$. Otherwise, there exists a unique block $A'\in\mathcal{B}(H)$ such that $A'\subseteq A$. Then we define $f_{H,G}(B)$ as follows: we define $f_{H,G}(A')=A$, and for every $B\in\mathcal{B}(H)\setminus\{A'\}$, we define $f_{H,G}(B)=B$. It is easy to verify that for all $B\in\mathcal{B}(H)$, $f_{H,G}(B)$ is well defined and $f_{H,G}(B)\in\mathcal{B}(G)$.
For $(u,B)\in\mathcal{P}_H$, we abuse notation and define that $f_{H,G}(u,B)=(u,f_{H,G}(B))$, and for $(B,u)\in\mathcal{P}_H$, we similarly define $f_{H,G}(B,u)=(f_{H,G}(B),u)$.
Let $C_H=C\cap V(H)$.
The following observation is easy to see.
\begin{observation}\label{obs:trunc}
Let $B\in\mathcal{B}(G)$. Then $0<|B\cap V(H)|<|B|$ if and only if $B=A$.
\end{observation}
\begin{lemma}\label{lem:nochange}
If $a\notin C$, $pot_G(C,(a,A))=0$ and $ua_G(a,A)=True$, then for each $p\in\mathcal{P}_H$, we have $ua_H(p)=ua_G(f_{H,G}(p))$ and $pot_H(C_H,p)=pot_G(C,f_{H,G}(p))$.
\end{lemma}
\begin{proof}
For ease of notation, we abbreviate $f_{H,G}$ to just $f$, $pot_G(C,p)$, where $p\in\mathcal{P}_G$, to just $pot_G(p)$, and $pot_H(C_H,p)$, where $p\in\mathcal{P}_H$, to just $pot_H(p)$ in this proof.

We first prove a claim that will be useful.
\begin{claim}\label{clm:missingblock}
Suppose that $A=\{a,b\}$. Then $ua_G(A,b)=False$ and $pot_G(A,b)=0$.
\end{claim}
It follows directly from Definition~\ref{def:ua} and the fact that $ua_G(a,A)=True$ that $ua_G(A,b)=False$. Then from Corollary~\ref{cor:pot}\ref{cor:xBu}, the fact that $a\notin C$ and $pot_G(a,A)=0$, we have that $pot_G(A,b)=0$. This proves the claim.
\medskip

We prove the lemma by induction on $d_H(p)$.
For the base case, suppose that $d_H(p)=0$.
Then we have $p=(B,u)$ for some $u\in V_{cut}(H)$ and $B\in\mathcal{B}_u(H)$.
We have by Definition~\ref{def:ua} that $ua_H(p)=ua_H(B,u)=True$. Also, we have by Corollary~\ref{cor:pot}\ref{cor:xBu} that $pot_H(B,u)=1-|B\cap\interior{C_H[B,u]}|$.
Suppose first that $B\not\subseteq A$.
Then we have by our definition of $f$ that $f(B,u)=(B,u)$; in particular $B\in\mathcal{B}(G)$ and $(B,u)\in\mathcal{P}_G$.
If $d_G(B,u)=0$, then we have $ua_G(f(p))=ua_G(B,u)=True$ and since $\kappa_G(B,u)=\emptyset$, we also have $V(G[B,u])=V(H[B,u])=B$. This implies that $\interior{C[B,u]}=\interior{C_H[B,u]}$.
Now by Corollary~\ref{cor:pot}\ref{cor:xBu}, we have that $pot_G(f(p))=pot_G(B,u)=1-|B\cap\interior{C[B,u]}|=1-|B\cap\interior{C_H[B,u]}|=pot_H(B,u)$. Thus we are done in this case.
So assume that $d_G(B,u)>0$. This means that $\kappa_G(B,u)\neq\emptyset$. As $\kappa_H(B,u)=\emptyset$, we know that there exists $v\in\kappa_G(B,u)$ such that $v\notin\kappa_H(B,u)$. Let $B'\in\mathcal{B}_v(G)\setminus\{B\}$. Note that $B'\notin\mathcal{B}(H)$; in fact, $B'\cap V(H)=\{v\}$. Since $|B'|\geq 2$, we have by Observation~\ref{obs:trunc} that $B'=A$. Note that this means that $\kappa_G(B,u)=\{v\}$ and that $\mathcal{B}_v(G)=\{B,B'\}$. If $|A|>2$, then $|B'\cap V(H)|\geq 2$, which means that $B'\cap V(H)$ is a block of $H$ in $\mathcal{B}_v(H)$. This contradicts the fact that $v\notin\kappa_H(B,u)$. So we can assume that $|A|=2$. We now have by Claim~\ref{clm:missingblock} that $pot_G(B',v)=0$ and $ua_G(B',v)=False$. Now we have by Definition~\ref{def:ua} that $ua_G(v,B)=False$ and by Corollary~\ref{cor:pot}\ref{cor:xuB} that $pot_G(v,B)=0$. Using Definition~\ref{def:ua} again, we get that $ua_G(B,u)=True$. Further, we now have from Corollary~\ref{cor:pot}\ref{cor:xBu} that $pot_G(B,u)=1-|B\cap\interior{C[B,u]}|=1-|B\cap\interior{C_H[B,u]}|=pot_H(B,u)$, and we are done.
Now consider the case when $B\subseteq A$. Then $f(B)=A$. Clearly, in this case, we have $|A|>2$, $B=A\setminus\{a\}$, and $f(B,u)=(A,u)$. Since $\kappa_H(B,u)=\emptyset$, we have $\kappa_G(A,u)=\{a\}$. As $ua_G(a,A)=False$, we have from Definition~\ref{def:ua} that $ua_G(A,u)=True$. Moreover, since $pot_G(a,A)=0$, we have from Corollary~\ref{cor:pot}\ref{cor:xBu} that $pot_G(A,u)=1-|A\cap\interior{C[A,u]}|$. As $a\notin C$, this implies that $pot_G(A,u)=1-|B\cap\interior{C_H[B,u]}|=pot_H(B,u)$. This completes the proof for the base case.

For the inductive step, we assume that for all $p'\in\mathcal{P}_H$ such that $d_H(p')<d_H(p)$, we have $ua_G(f(p'))=ua_H(p')$ and $pot_G(f(p'))=pot_H(p')$.
Suppose first that $p=(u,B)$ for some $u\in V_{cut}(H)$ and $B\in\mathcal{B}_u(H)$. Let us first consider the case when
$|\beta_H(u,B)|=|\beta_G(f(u,B))|$ (i.e. it is not the case that $A\in\beta_G(u,B)$ and $|A|=2$). Then for each $X\in\beta_H(u,B)$, we have by the induction hypothesis that $ua_H(X,u)=ua_G(f(X,u))$ and $pot_H(X,u)=pot_G(f(X,u))$. Then it follows from Corollary~\ref{cor:pot}\ref{cor:xuB} that $pot_G(f(u,B))=pot_H(u,B)$. Similarly, it follows from Definition~\ref{def:ua} that $ua_G(f(u,B))=ua_H(u,B)$.
Next, suppose that $A\in\beta_G(u,B)$ and $|A|=2$. Then by Claim~\ref{clm:missingblock}, we have that $ua_G(A,u)=False$ and $pot_G(A,u)=0$. Now it follows from the induction hypothesis and Definition~\ref{def:ua} that $ua_G(f(p))=ua_H(p)$. Similarly, using the induction hypothesis and Corollary~\ref{cor:pot}\ref{cor:xuB}, we get $pot_G(f(p))=pot_H(p)$.

Finally, we consider the case when $p=(B,u)$ for some $u\in V_{cut}(H)$ and $B\in\mathcal{B}_u(H)$. By the induction hypothesis, we know that for each $v\in\kappa_H(B,u)$, $ua_H(v,B)=ua_G(f(v,B))$ and $pot_G(f(v,B))=pot_H(v,B)$.
Suppose that $B\not\subseteq A$. In this case, we have $f(B,u)=(B,u)$, and also that $\kappa_G(B,u)=\kappa_H(B,u)$. It is easy to see that $B=\kappa_H(B,u)\cup\{u\}$ if and only if $B=\kappa_G(B,u)\cup\{u\}$. We now have by Definition~\ref{def:ua} and the induction hypothesis that $ua_G(f(B,u))=ua_G(B,u)=ua_H(B,u)$.
Since we have $B\cap\interior{C_H[B,u]}=B\cap\interior{C[B,u]}$,
we now have by Corollary~\ref{cor:pot}\ref{cor:xBu} that $pot_G(f(B,u))=pot_H(B,u)$.
Finally, we consider the case when $B\subseteq A$. Then we have $\kappa_G(A,u)=\kappa_H(B,u)\cup\{a\}$. It is easy to see that $B=\kappa_H(B,u)\cup\{u\}$ if and only if $A=\kappa_G(A,u)\cup\{u\}$. Since we know that $ua_G(a,A)=True$, we now have by Definition~\ref{def:ua} that $ua_G(f(B,u))=ua_G(A,u)=ua_H(B,u)$.
Since $a\notin C$, we know that $B\cap\interior{C_H[B,u]}=A\cap\interior{C[A,u]}$.
Now, as $pot_G(a,A)=0$, we have by Corollary~\ref{cor:pot}\ref{cor:xBu} that $pot_G(f(B,u))=pot_G(A,u)=pot_H(B,u)$.
\end{proof}
\section{Putting the pieces together}\label{sec:main}
\begin{lemma}\label{lem:norigidtoken}
Let $G$ be a connected block graph and $C$ be an independent set of $G$.
Let $H$ be a connected component of the graph obtained by removing all vertices of $G$ that are rigid for $C$. There are no vertices in $H$ that are rigid for $C\cap V(H)$.
\end{lemma}
\begin{proof}
Let $R$ denote the set of vertices in $G$ that are rigid for $C$. Since $H$ is a connected component of $G-R$, it must be the case that $N_G(H)\subseteq R$. Let $N_G(H)=\{r_1,r_2,\ldots, r_k\}$.
We now define connected induced subgraphs $G_0,G_1,G_2,\ldots,G_k$ of $G$, each containing $H$ as a subgraph, such that for each $i\in\{0,1,\ldots,k-1\}$, $V(H)\subseteq V(G_i)$ and $\{r_{i+1},r_{i+2},\ldots,r_k\}\subseteq V(G_i)$.
Let $G_0=G$.
For $i\in\{1,2,\ldots,k\}$, we define $G_i$ inductively, assuming that $G_0,G_1,\ldots,G_{i-1}$ have been defined. By the induction hypothesis, we know that $r_i,r_{i+1},\ldots,r_k\in V(G_{i-1})$.
Then by Observation~\ref{obs:blockgraph}, there exists a unique block $B_i\in\mathcal{B}_{r_i}(G_{i-1})$ such that $B_i\cap V(H)\neq\emptyset$ (since $r_i\in N_G(H)\cap V(G_{i-1})$ and $V(H)\subseteq V(G_{i-1})$). We define $G_i=G_{i-1}-V(G_{i-1}[r_i,B_i])$. It is easy to see that $G_i$ is an induced subgraph of $G_{i-1}$, and therefore also of $G$. Moreover, since $G_{i-1}$ is connected, it follows that $G_i$ is also connected. By Observation~\ref{obs:blockgraph}, we have that that $V(H)\subseteq V(G_i)$, and also that no vertex in $N_G(H)$ is in $V(G_{i-1}[r_i,B_i]\setminus\{r_i\}$. This implies that $r_{i+1},r_{i+2},\ldots,r_k\in V(G_i)$, as required.

Note that $V(H)\subseteq V(G_k)$, $G_k$ is a connected graph, and $N_G(H)\cap V(G_k)=\emptyset$. This implies that $G_k=H$.
Since for each $i\in\{1,2,\ldots,k\}$, we have that $G_i=V(G_{i-1})-V(G_{i-1}[r_i,B_i]$ (recall that $B_i\in\mathcal{B}_{r_i}(G_{i-1})$), we can define the function $f_{G_i,G_{i-1}}$ as described in the beginning of Section~\ref{sec:removerigid}. For ease of notation, we abbreviate $f_{G_i,G_{i-1}}$ to just $f_i$, and $pot_{G_i}(C\cap V(G_i),p)$, where $p\in\mathcal{P}_{G_i}$, to just $pot_i(p)$ for the remainder of this proof.
\begin{claim}
For each $i\in\{1,2,\ldots,k\}$ and $p\in\mathcal{P}_{G_i}$, $pot_i(p)=pot_{i-1}(f_i(p))$ and $ua_{G_i}(p)=ua_{G_{i-1}}(f_i(p))$.
\end{claim}
We prove this by induction on $i$. As the base case, we consider the case when $i=1$. As $G_0=G$ and $r_1\in R$, we have by Lemma~\ref{lem:rigidnotoken} that $r_1\notin C$, and by Observation~\ref{obs:rigidua} that $pot_0(r_1,B_1)=0$ and $ua_G(r_1,B_1)=True$. Now since $G_1=G-V(G[r_1,B_1])$, we have by Lemma~\ref{lem:nochange} that for any $p\in\mathcal{P}_{G_1}$, $ua_{G_1}(p)=ua_G(f_1(p))$ and $pot_1(p)=pot_0(f_1(p))$.
For the inductive step, we assume that the claim holds for all $j\in\{1,2,\ldots,k\}$ such that $j<i$. Recall that $G_i=G_{i-1}-V(G_{i-1}[r_i,B_i])$. 
Notice that $(r_i,B_i)\in\mathcal{P}_{G_{i-1}}$. By the induction hypothesis, we know that $pot_{i-1}(r_i,B_i)=pot_{i-2}(f_{i-1}(r_i,B_i))=pot_{i-2}(r_i,f_{i-1}(B_i))=pot_{i-3}(f_{i-2}(r_i,f_{i-1}(B_i)))=pot_{i-3}(r_i,f_{i-2}\circ f_{i-1}(B_i))=\cdots=pot_0(r_i,f_1\circ f_2\circ\cdots\circ f_{i-1}(B_i))$. By the same argument, we also get that $ua_{G_{i-1}}(r_i,B_i)=ua_G(r_i,f_1\circ f_2\circ\cdots\circ f_{i-1}(B_i))$. It is easy to see that $f_1\circ f_2\circ\cdots\circ f_{i-1}(B_i)\in\mathcal{B}_{r_i}(G)$. Since $r_i\in R$, we then have by Observation~\ref{obs:rigidua} that $pot_0(r_i,f_1\circ f_2\circ\cdots\circ f_{i-1}(B_i))=0$ and $ua_G(r_i,f_1\circ f_2\circ\cdots\circ f_{i-1}(B_i))=True$, which implies that $pot_{i-1}(r_i,B_i)=0$ and $ua_{G_{i-1}}(r_i,B_i)=True$. Observe that we have from Lemma~\ref{lem:rigidnotoken} that $r_i\notin C$, which implies that $r_i\notin C\cap V(G_{i-1})$. Now it follows from Lemma~\ref{lem:nochange} that for each $p\in\mathcal{P}_{G_i}$, $pot_i(p)=pot_{i-1}(f_i(p))$ and $ua_{G_i}(p)=ua_{G_{i-1}}(f_i(p))$. This completes the proof of the claim.
\medskip

From the above claim, we have that for each $p\in\mathcal{P}_{G_k}$, $pot_k(p)=pot_0(f_1\circ f_2\circ\cdots\circ f_k(p))$ and $ua_{G_k}(p)=ua_G(f_1\circ f_2\circ\cdots\circ f_k(p))$. Suppose for the sake of contradiction that there is a vertex $u$ that is rigid in $H$ for $C\cap V(H)$. Then there exists $B,B'\in\mathcal{B}_u(H)$ such that $pot_k(u,B)=pot_k(u,B')=0$ and $ua_{G_k}(u,B)=ua_{G_k}(u,B')=True$ (recall that $H=G_k$). As observed above, this means that $pot_0(u,f_1\circ f_2\circ\cdots\circ f_k(B))=pot_0(u,f_1\circ f_2\circ\cdots\circ f_k(B'))=0$ and $ua_G(u,f_1\circ f_2\circ\cdots\circ f_k(B))=ua_G(u,f_1\circ f_2\circ\cdots\circ f_k(B'))=True$. This implies that $u$ is rigid in $G$ for $C$, or in other words, $u\in R$. This is a contradiction since $u\in V(H)$ and $V(H)\subseteq V(G)\setminus R$.
\end{proof}

\begin{lemma}\label{lem:norigid}
Let $G$ be a connected block graph and $C_1,C_2$ be independent sets of $G$. If no cut-vertex of $G$ is rigid for $C_1$ or $C_2$, and $|C_1|=|C_2|$, then $C_2$ is reachable from $C_1$.
\end{lemma}
\begin{proof}
We prove this lemma using induction on $|C_1|=|C_2|=r$. As the base case, we assume that $r=1$. Let $C_1=\{u\}$ and $C_2=\{v\}$. Suppose that $P=w_1w_2\ldots w_k$ is a path from $u$ to $v$ in $G$ such that $w_1=u$ and $w_k=v$. Since $G$ is connected, $P$ is guaranteed to exist. Let $D_1=C_1$. For $i\in\{2,3,\ldots,k\}$, we define $D_i=(D_{i-1}\setminus \{w_{i-1}\})\cup\{w_i\}$. Clearly, $C_1=D_1\rightarrow D_2\rightarrow \cdots \rightarrow D_k=C_2$, and hence $C_2$ is reachable from $C_1$. This proves the base case. 

For the inductive step, we assume that if no cut-vertex of $G$ is rigid for independent sets $C'_1$ and $C'_2$ of $G$, and $|C'_1|=|C'_2|<r$, then $C'_2$ is reachable from $C'_1$. Now, let us prove the lemma for $C_1$ and $C_2$.
We claim that there exists $l\in V_{cut}(G)$ such that there exists at most one $B\in\mathcal{B}_l(G)$ having $\kappa_G(B,l)\neq\emptyset$.
Indeed, we can choose as $l$ an endpoint of a longest path in $G$ whose both endpoints are cut-vertices of $G$. Choose as $L$ a block from $\mathcal{B}_l(G)\setminus\{B\}$. Note that $G[L,l]$ is a complete graph. Let $h\in L\setminus \{l\}$. First, we define independent sets $S_1, S_2$ of $G$ such that $h \in S_1\cap S_2$, $S_1$ is reachable from $C_1$, and $S_2$ is reachable from $C_2$.
\medskip

Let $i\in\{1,2\}$.
\medskip

If $h\in C_i$, then we simply define $S_i=C_i$ and we are done. So we assume that $h\notin C_i$. Since $N_G(h)\subseteq L$, we know that $|N_G(h)\cap C_i|\leq 1$. Thus, if there exists $v\in N_G(h)\cap C_i$, then $S_i=(C_i\setminus\{v\})\cup\{h\}$ is an independent set of $G$ such that $C_i\rightarrow S_i$, and we are done.
So we assume that $(\{h\}\cup N_G(h))\cap C_i=\emptyset$, or in other words, every vertex in $C_i$ is at a distance of at least 2 from $h$ in $G$.
Choose a vertex $v\in C_i$ for which the distance between $v$ and $h$ in $G$ is as small as possible. Let $P=w_1w_2\ldots w_k$ be the shortest path from $v$ to $h$ in $G$ such that $w_1=v$ and $w_k=h$ (thus the distance between $v$ and $h$ in $G$ is $k-1$). Note that by our assumption, we have $k\geq 3$. Since $G$ is a block graph, it can be easily seen that $\{w_2,w_3,\ldots, w_{k-1}\}\subseteq V_{cut}(G)$, and also that $w_{k-1}=l$. If for some $j\in\{3,4,\ldots, k\}$, there exists a vertex $x\in N_G(w_j)\cap C_i$, then 
$xw_jw_{j+1}\ldots w_k$ is a path in $G$ between $x$ and $h$ having length $k-j+1$. Then $x$ is a vertex in $C_i$ such that the distance between it and $h$ in $G$ is smaller than the distance between $v$ and $h$ in $G$, which contradicts the choice of $v$. Therefore, 
$w_j$ is not under attack in $C_i$, or in other words $N_G(w_j)\cap C_i=\emptyset$, for any $j\in\{3,4,\ldots, k\}$.

Let $X\in\mathcal{B}(G)$ such that $X$ contains the edge $w_2w_3$. Notice that $w_2$ is under attack in $C_i[w_2,X]$ (as $w_1\in C_i$). By Lemma~\ref{lem:nonrigid}, there exists an independent set $D$ that is $(w_2,X)$-reachable from $C_i$ such that $|N_G(w_2)\cap D[w_2,X]|=1$. 
Let $N_G(w_2)\cap D[w_2,X]=\{y\}$. Observe that $yw_2\ldots w_k$ is a path from $y$ to $h=w_k$ in $G$. As $D[X,w_2]=C_i[X,w_2]$, we have that $N_G(w_j)\cap D=\emptyset$ for each $j\in\{3,4,\ldots,k\}$. Let $D_1=D$ and $D_2=(D_1\setminus \{y\})\cup\{w_2\}$. For $j\in\{3,4,\ldots,k\}$, we define $D_j=(D_{j-1}\setminus \{w_{j-1}\})\cup\{w_j\}$. It follows from the observations above that each of $D_1,D_2,\ldots,D_k$ is an independent set of $G$, and we also clearly have $D=D_1\rightarrow D_2\rightarrow \cdots \rightarrow D_k$. We define $S_i=D_k$. Then $S_i$ is reachable from $D$, and hence also from $C_i$. We also have $h=w_k\in D_k=S_i$.

Recall that $l=w_{k-1}$. Thus $l\in D_{k-1}$, and since $D_{k-1}\rightarrow D_k=S_i$, we have that $N_G(l)\cap S_i=\{h\}$. This implies that $S_i[l,B]=\{h\}$ (recall that $B\in \mathcal{B}_l(G)$ such that $\kappa_G(B,l)\ne\emptyset$). Let $H=G-V(G[l,B])$. Clearly, $H$ is a connected block graph. Note that we can define a function $f_{H,G}:\mathcal{P}_H\rightarrow\mathcal{P}_G$ as described in the beginning of Section~\ref{sec:removerigid}. 
Since each $B'\in\beta_G(l,B)$ is a complete graph, we have by Observation~\ref{obs:capcomplete} that $ua_G(B',l)=True$, and so in particular $ua_G(L,l)=True$.
Since $h\in S_i$, this implies by Corollary~\ref{cor:pot}\ref{cor:xBu} that $pot_G(S_i,(L,l))=0$.
Moreover, since $S_i[l,B]=\{h\}$, we have that for each $B'\in\beta_G(l,B)\setminus\{L\}$, $|\interior{S_i[B',l]}\cap B'|=0$, which implies by Corollary~\ref{cor:pot}\ref{cor:xBu} that $pot_G(S_i,(B',l))=1$. Further, we get by Definition~\ref{def:ua} that $ua_G(l,B)=True$, which implies by Corollary~\ref{cor:pot}\ref{cor:xuB} that $pot_G(S_i,(l,B))=0$. 
Therefore, since $l\notin S_i$, we get by Lemma~\ref{lem:nochange} that for each $p\in\mathcal{P}_H$, $ua_H(p)=ua_G(f_{H,G}(p))$ and $pot_H(S_i\cap V(H),p)=pot_G(S_i,f_{H,G}(p))$.

Since $S_i$ is reachable from $C_i$ and no vertex is rigid for $C_i$, by Lemma~\ref{lem:rigidmatch} it is clear that  there does not exist a vertex that is rigid for $S_i$. This implies that for each $u\in V_{cut}(G)$, there does not exist $B',B''\in\mathcal{B}_u(G)$ such that $pot_G(S_i,(B',u))=pot_G(S_i,(B'',u))=0$ and $ua_G(B',u)=ua_G(B'',u)=True$.
From our previous observation that for each $p\in\mathcal{P}_H$, $ua_H(p)=ua_G(f_{H,G}(p))$ and $pot_H(S_i\cap V(H),p)=pot_G(S_i,f_{H,G}(p))$,
we get that for each $u\in V_{cut}(H)$, there does not exist $B',B''\in\mathcal{B}_u(H)$ such that $pot_H(S_i\cap V(H),(B',u))=pot_H(S_i\cap V(H),(B'',u))=0$ and $ua_H(B',u)=ua_H(B'',u)=True$. This implies that in the graph $H$, no cut-vertex is rigid for the independent set $S_i\cap V(H)$ of $H$.
\medskip

We thus have that no cut-vertex of $H$ is rigid for either of the independent sets $S_1\cap V(H)$ or $S_2\cap V(H)$ of $H$.
Further, since $|S_1|=|C_1|=|C_2|=|S_2|=r$ and $S_1[l,B]=S_2[l,B]=\{h\}$, we have $|S_1\cap V(H)|=|S_2\cap V(H)|=r-1$, which implies by the induction hypothesis that in $H$, $S_2\cap V(H)$ is reachable from $S_1\cap V(H)$. This implies that there exist independent sets $R_0,R_1,\ldots,R_q$ of $H$, where $q\geq 0$, such that $S_1\cap V(H)=R_0\rightarrow R_1\rightarrow\cdots\rightarrow R_q=S_2\cap V(H)$. For each $j\in\{0,1,\ldots,q\}$, since $N_G(h)\cap V(H)=\emptyset$, it is clear that $R'_j=R_j\cup\{h\}$ is an independent set of $G$. It is easy to see that for each $j\in\{0,1,\ldots,q-1\}$, since $R_j\rightarrow R_{j+1}$, we also have $R'_j\rightarrow R'_{j+1}$. Moreover, as $S_1[l,B]=S_2[l,B]=\{h\}$, we have $R'_1=S_1$ and $R'_q=S_2$. Hence, we get $S_1=R'_0\rightarrow R'_1\rightarrow\cdots\rightarrow R'_q=S_2$, which shows that $S_2$ is reachable from $S_1$. Since we have already showed that $S_1$ is reachable from $C_1$ and $S_2$ is reachable from $C_2$, we get that $C_2$ is reachable from $C_1$. Hence, we are done.     
\end{proof}
\begin{theorem}\label{thm:correctness}
Let $G$ be a block graph and let $C_1,C_2$ be independent sets of $G$. Let $W_1,W_2$ be the set of vertices of $G$ that are rigid for $C_1,C_2$ respectively. Then $C_1$ is reachable from $C_2$ if and only if $W_1=W_2$, and for every connected component $H$ of $G-W_1$, we have $|C_1\cap V(H)|=|C_2\cap V(H)|$.
\end{theorem}
\begin{proof}
    First, we prove the forward direction for which we consider the contrapositive statement. Suppose either $W_1\ne W_2$ or there exists a connected component $H$ of $G-W_1$ such that $|C_1\cap V(H)|\ne |C_2\cap V(H)|$. If $W_1\ne W_2$, then by Lemma~\ref{lem:rigidmatch} we have that $C_1$ is not reachable from $C_2$ and we are done in this case. Hence, 
    we assume that $W_1=W_2$ and that there exists a connected component $H$ of $G-W_1$ such that $|C_1\cap V(H)|\ne |C_2\cap V(H)|$. Suppose for the sake of contradiction that $C_1$ is reachable from $C_2$. Then, there exists independent sets $D_0,D_1\ldots,D_k$, where $k\geq 0$, such that $C_1=D_0\rightarrow D_1\rightarrow\cdots\rightarrow D_k=C_2$. 
    Let $j=\min\{i\in\{0,1,\ldots,k\}\colon |D_i\cap V(H)|\neq |C_1\cap V(H)|\}$.
    Then clearly, we have $j>0$ and $|D_{j-1}\cap V(H)|=|C_1\cap V(H)|$. Since $D_{j-1}\rightarrow D_j$, there exists $uv\in E(G)$ such that $(D_{j-1}\setminus D_j)\cup (D_j\setminus D_{j-1})=\{u,v\}$. As $|D_j\cap V(H)|\neq |D_{j-1}\cap V(H)|$, it must be the case that one of $u,v$ is in $V(H)$ and the other is not in $V(H)$. Without loss of generality, assume that $u\in V(H)$ and $v\notin V(H)$. Since $uv\in E(G)$, it follows that $v\in N_G(H)$. Since $H$ is a connected component of $G-W_1$, we now have that $v\in W_1$, or in other words, $v$ is rigid for $C_1$. This means that one $D_{j-1}$, $D_j$ is an independent set of $G$ that is reachable from $C_1$ that contains a vertex that is rigid for $C_1$. This contradicts Lemma~\ref{lem:rigidnotoken}.

    Next, we prove the reverse direction. Suppose that $W_1=W_2$ and for every connected component $H$ of $G-W_1$, we have $|C_1\cap V(H)|=|C_2\cap V(H)|$. By Lemma~\ref{lem:norigidtoken}, we have that each connected component $H$ of $G-W_1$ has no vertex that is rigid for $C_1 \cap V(H)$ or $C_2 \cap V(H)$. For each connected component $H$ of $G-W_1$, since $|C_1\cap V(H)|=|C_2\cap V(H)|$, applying Lemma~\ref{lem:norigid} to $H$ and the independent sets $C_1 \cap V(H)$ and $C_2\cap V(H)$ of $H$, we can conclude that $C_1\cap V(H)$ is reachable from $C_2 \cap V(H)$ in $H$. Moreover, by Lemma~\ref{lem:rigidnotoken}, we know that $W_1\cap C_1=W_1\cap C_2=\emptyset$, which implies that $C_1,C_2\subseteq V(G-W_1)$. It now follows from Observation~\ref{obs:disconnected} that $C_1$ is reachable from $C_2$ in $G-W_1$, and then from Observation~\ref{obs:subgraph} that $C_1$ is reachable from $C_2$ in $G$.
\end{proof}
\begin{corollary}\label{cor:decision}
\prob\ is polynomial time solvable on block graphs.
\end{corollary}
\begin{proof}
Consider an input instance $(G,C_1,C_2)$ of the \prob\ problem, where $G$ is a block graph and $C_1,C_2$ are independent sets of $G$.
We can assume that the input block graph $G$ is connected, as otherwise, we can simply solve the problem on each connected component of $G$ (Observation~\ref{obs:disconnected}). Note that this means that $|V(G)|\leq |E(G)|$.
We first run {\sc Compute-potentials}$(G,C_1)$ and {\sc Compute-potentials}$(G,C_2)$ to compute $pot_G(C_1,p)$ and $pot_G(C_2,p)$ for all $p\in\mathcal{P}_G$. By Theorem~\ref{thm:polytime}, this can be done in $O(|E(G)|^4)$ time. It is easy to see that once the potentials with respect to $C_1$ and $C_2$ have been computed, the set of vertices that are rigid for $C_1$ and the set of vertices that are rigid for $C_2$ can both be determined, and the equality of these sets checked, in $O(|V(G)|+|E(G)|)=O(|E(G)|)$ time. If the set of vertices that are rigid is not the same for both $C_1$ and $C_2$, we can return ``No'' since we know by Theorem~\ref{thm:correctness} that $C_2$ is not reachable from $C_1$. Otherwise, we remove the (common) set of rigid vertices for both $C_1$ and $C_2$ to obtain a graph $G'$ --- this can be done in $O(|V(G)|+|E(G)|)=O(|E(G)|)$ time. Clearly, in $O(|E(G)|)$ time, we can determine the connected components of $G'$ and also determine if $|C_1\cap V(H)|=|C_2\cap V(H)|$ for each connected component $H$ of $G'$. We return ``No'' if there exists some connected component $H$ of $G'$ such that $|C_1\cap V(H)|\neq |C_2\cap V(H)|$, and otherwise we return ``Yes''. From Theorem~\ref{thm:correctness}, it follows that our algorithm returns ``Yes'' if and only if the independent set $C_2$ of $G$ is reachable from $C_1$. Clearly, this algorithm runs in time $O(|E(G)|^4)$. This completes the proof.
\end{proof}
\section{Conclusion}
Although Corollary~\ref{cor:decision} tells us that there is a polynomial time algorithm that can determine whether two independent sets in a block graph are reconfigurable with each other, that algorithm does not produce a reconfiguration sequence between the two independent sets. In fact, it is not even clear whether there is a reconfiguration sequence of polynomial length between any two independent sets in a block graph that are reconfigurable with each other. It was shown in Bria\'nski et al.~\cite{ts_interval} that for any two independent sets that are reconfigurable with each other in an interval graph, there is a reconfiguration sequence of length $O(n^3)$ between them (where $n$ is the number of vertices in the graph). For trees, it is shown in Demaine et al.~\cite{DEMAINE2015132tree} that any two independent sets that are reconfigurable with each other have a reconfiguration sequence of length $O(n^2)$ between them.

\bibliographystyle{plain}
\bibliography{ref}

\end{document}